\documentclass[12pt,letter]{article}

\usepackage{latexsym,amsmath,amssymb,amsthm,amsfonts,bbm,graphicx,natbib,breqn,epsfig,bm,authblk,url,thmtools,xcolor}

\usepackage{float} 
\usepackage{booktabs} 
\usepackage{graphicx} 
\usepackage[margin=1cm]{caption} 
\usepackage{mathtools}
\usepackage{subcaption}

\floatstyle{plain}
\newfloat{Algorithm}{thp}{lop}
\floatname{Algorithm}{Algorithm}

\declaretheoremstyle[
notefont=\bfseries, notebraces={}{},
bodyfont=\normalfont,
postheadspace=0.5em,
numbered=yes,
]{mystyle}
\declaretheorem[style=mystyle]{lemma}

\theoremstyle{definition}

    \def \mC {\text{\boldmath$C$}}

\def \svec {\text{\boldmath$s$}}

\def \wvec {\text{\boldmath$w$}}    \def \mW {\text{\boldmath$W$}}
    
\def \yvec {\text{\boldmath$y$}}    \def \mY {\text{\boldmath$Y$}}
\def \zvec {\text{\boldmath$z$}}

\def \alphavec        {\text{\boldmath$\alpha$}}
\def \betavec         {\text{\boldmath$\beta$}}
\def \gammavec        {\text{\boldmath$\gamma$}}

\def \epsilonvec      {\text{\boldmath$\epsilon$}}

\def \zetavec         {\text{\boldmath$\zeta$}}
\def \etavec          {\text{\boldmath$\eta$}}
\def \thetavec        {\text{\boldmath$\theta$}}

\def \lambdavec       {\text{\boldmath$\lambda$}}
\def \muvec           {\text{\boldmath$\mu$}}

\def \xivec           {\text{\boldmath$\xi$}}

\def \rhovec          {\text{\boldmath$\rho$}}

\def \varphivec       {\text{\boldmath$\varphi$}}

\def \mSigma   {\mathbf{\Sigma}}

\usepackage{algorithm}
\usepackage{algorithmic}
\begin{document}

\title{Variational inference for cutting feedback in misspecified models}
\date{\empty}

\author[1]{Xuejun Yu}
\author[1,2]{David J. Nott\thanks{Corresponding author:  standj@nus.edu.sg}}
\author[3]{Michael Stanley Smith}
\affil[1]{Department of Statistics and Data Science, National University of Singapore, Singapore 117546}
\affil[2]{Institute for Operations Research and Analytics, National University of Singapore, Singapore 119077}
\affil[3]{Melbourne Business School, University of Melbourne, 200 Leicester Street,
Carlton VIC 3053}

\maketitle

\vspace{-0.3in}
\begin{abstract}
\noindent Bayesian analyses combine information represented
by different terms in a joint Bayesian model.  When one or more of the terms is misspecified, it can be
helpful to restrict the use of information from suspect model components to modify posterior inference.
This is called ``cutting feedback", and both the specification
and computation of
the posterior for such ``cut models'' is challenging.   
In this paper, we define cut posterior distributions as solutions
to constrained optimization problems, and propose
variational methods
for their computation.
These methods are faster than existing Markov chain Monte Carlo (MCMC) approaches by an order of magnitude. 
It is also shown that variational methods allow 
for the evaluation of computationally intensive conflict checks
that can be used to decide whether or not feedback should be cut. 
Our methods are illustrated in a number of simulated and real examples, including an application
where recent methodological advances that
combine variational inference and MCMC within the variational optimization
are used.


\end{abstract}

\section{Introduction}\label{sec:Intro}


Bayesian inference combines information represented by different
terms in a joint Bayesian model. For example, these terms may be likelihood functions for different data sources, or components
of hierarchical prior distributions for model parameters.  
Sometimes the terms
can be grouped into subsets, where
each subset is called a ``module'' that represents
one aspect of the model.
Often different sources of data can be
used to estimate the parameters of each module
separately.  However, it is the joint model that
specifies how all these data sources can be combined to allow joint Bayesian inference with uncertainty propagation between modules.

When one of the modules is misspecified, 
it may be desirable to restrict the way it interacts with other modules
to produce estimates that deviate from the full Bayesian posterior distribution. 
One approach to achieve this is to start
with a conditional or sequential representation of the joint posterior distribution, but then
modify some of its terms to ensure that suspect
information is removed in inference for some parameters.
This is referred to as modularization, where the modules interact more weakly than
in a conventional Bayesian analysis~\citep{liu+bb09}.  

In this paper we consider the computational problems which arise in a 
form of modularization called ``cutting feedback'' and propose optimization-based
variational approximation methods to formulate cut models.  We do this in two ways. First, similar to
Markov chain Monte Carlo (MCMC) based cut methods which modify Gibbs sampling approaches, we consider modified variational
inference methods based on deletion of so-called ``messages'' within variational message passing algorithms.  
Second, we consider an explicit formulation of the cut posterior distribution as the solution to a constrained optimization
problem, and approximate it using black box variational inference methods. In this case, our variational inference methods are faster than MCMC approaches by an order of magnitude.  We also
consider how variational inference can be used in computationally burdensome checks for conflicting information, and
these checks are helpful for deciding whether to cut in the first place.  

Much early work on modularization arose in 
pharmacokinetic/pharmacodynamic modeling applications, where
a cutting feedback approach to Bayesian analyses was motivated 
as an errors-in-variables method \citep{bennett+w01,lunn+bsgn09}.  
In these applications, there are two modules: a pharmacokinetic (PK) model and a pharmacodynamic (PD) model.
The PD model describes the effect of drug concentration
on a physiological outcome, while the PK model describes 
the evolution of drug concentration in the bloodstream. 
Specification of a realistic PD model is difficult, so
this module may be misspecified, while there is greater confidence in the PK model. 
There are separate data sources which
can be used to infer the PK and PD model parameters, and it is desirable to not let the possibly misspecified PD model contaminate estimation of the 
PK model parameters, while still allowing appropriate propagation of uncertainty.    
Cutting feedback methods in PK/PD modelling were
inspired by earlier two-stage frequentist methods 
\citep{zhang+bs03,zhang+bs03b}.  
In other complex applications, there can be similar concerns that
suspect modules may contaminate estimation of parameters of interest.  
\citet{jacob+mhr17} give a recent
overview of modularized Bayesian analyses with a decision-theoretic perspective, including references to applications 
in areas such as climate modelling, epidemiology, causal
inference using propensity scores and meta-analysis among others.  
  
Modularization raises interesting statistical issues, and 
\citet{liu+bb09} consider these in the context of analyzing computer models  \citep{kennedy+o01}.  They 
discuss the motivations for using modified Bayesian inference within an existing
flawed model, rather than the conventional approach of Bayesian model criticism leading to model improvement.  
As well as providing more appropriate inferences, modularization can ensure 
that model parameters retain their intended meaning within suspect modules, assisting 
model interpretation and criticism.   \citet{lunn+bsgn09} motivate 
cut procedures as corresponding to specification of ``distributional constants" and 
they argue that it is sometimes reasonable to consider a cut posterior specified through inconsistent conditional distributions.
An alternative to the dichotomy of using the cut or full posterior
distribution has recently been considered by \cite{carmona+n20}, where they outline
a semi-modular method in which feedback is partially cut.  
\cite{nicholls+lwc22} consider the justification of semi-modular inference
from a generalized Bayesian perspective.  

Our objective in this paper is to describe the useful role that variational inference can play in the analysis of cut models.
In concurrent independent work, \cite{carmona+n21} also consider
variational inference for modularized
Bayesian analyses, and we discuss the ways that their contribution
differs from ours in Section 4.  
In Section 2, we describe some of the ways that cut posterior distributions are defined in the existing literature.
The definition can be implicit through
modification of an MCMC algorithm targeting the full posterior, or explicit through direct specification of a target
distribution, and we discuss both perspectives.  We also
survey some of the many applications of
cutting feedback methods.  
In Section 3, we give a brief introduction to variational inference  
and then describe cut procedures based on mean field variational approximations.  
In this context, variational message passing algorithms provide a natural way to define
a variational cut posterior implicitly, similar to MCMC implementations of cutting feedback 
based on modified Gibbs sampling algorithms.  
In Section 4, we describe why conventional Bayesian 
computation using MCMC is difficult
for cut models. A simplified two module system is then considered 
where an explicit formulation of the cut posterior distribution is available,
and can be expressed as the solution to a constrained optimization problem.  
This demonstrates 
that the cut posterior distribution is a variational approximation to the full posterior distribution
for a certain approximating family, and
motivates the use of fixed form variational approximations for computation.
Section 5 describes the use of prior-data conflict checking methods for 
deciding
whether or not to cut. Here, variational inference  
greatly facilitates a practical computational implementation of the methods.  
Section 6 illustrates the methodology for
two real data examples discussed in the literature previously.
In particular, we consider an example from~\cite{styring17} and \cite{carmona+n20} where we use a recently 
developed method~\citep{loaiza-maya+snd20} that combines MCMC and variational
approximation within the variational optimization. The approach allows
for the imputation of an unobserved discrete variable, which is otherwise difficult to
do within the variational optimization. 
Section 7 gives some concluding discussion.

\section{Cutting feedback}\label{cutting}

In this section we discuss the ways that cut posterior distributions are usually defined in the existing
literature. This includes implicit and explicit definitions. However, 
before doing so it is helpful to consider a simple motivating example from \citet{liu+bb09}, where the full posterior distribution 
behaves in undesirable ways and where cut procedures are beneficial.  

\subsection{Illustrative example}\label{sec:simpleill}
Suppose we have a small sample $\zvec=(z_1,\dots, z_{n_1})^\top$ with
$z_i\sim N(\varphi,1)$, and we are interested in inference about $\varphi$.  The prior distribution for 
$\varphi$ is $N(0,\delta_1^{-1})$, where $\delta_1>0$ is the prior precision.  Due to the small sample size $n_1$, 
it is thought desirable to consider another source of data $\wvec=(w_1,\dots, w_{n_2})^\top$, 
for which $w_i\sim N(\varphi+\eta,1)$.  The sample size $n_2$ is large, but the mean of $w_i$ is equal to 
$\varphi+\eta$ rather than the parameter of interest $\varphi$, so that
$\eta$ is a bias parameter with prior
$\eta\sim N(0,\delta_2^{-1})$, where $\delta_2>0$ is the prior precision. Suppose that the analyst
has high confidence that the bias $\eta$ is small, and uses a large value for $\delta_2$,
resulting in a prior density for $\eta$ concentrated around $0$.  
Then if the true bias is in fact large, the information from the biased sample $\wvec$ can dominate inference
about $\varphi$ and furthermore the strong prior on $\eta$ can result in misleading inferences.
In this case, \cite{liu+bb09} point out that there is little to gain
from using the biased data for inference about
$\varphi$.  The model can be considered as a two module system.  One module contains the prior for $\varphi$
and the likelihood term for $\zvec$, and another module contains the prior for $\eta$ and the likelihood
term for $\wvec$.  The misspecified module is the second one, and it is the prior term for $\eta$ that
introduces posterior inaccuracy.

To illustrate the sizable impact of the misspecified module,   
we simulate $n_1=100$, $n_2=1000$ observations from the data generating
process with parameter values $\varphi=0$ and $\eta=1$.  
The prior precision parameters are $\delta_1=1$ and $\delta_2=100$, with the latter 
chosen so that the true value $\eta=1$ lies out in the tails of the prior.  Figure \ref{biased-marginal} in Section 3.3 shows the poor
behaviour of the full posterior distribution in this example, and compares this with a ``cut'' posterior distribution 
where the influence of the biased data is removed in inference about $\varphi$.  
The cut posterior inferences are more reasonable than those from the full posterior.  
This simple example demonstrates the potential advantages of cutting feedback, 
and the way that misspecification of one module can contaminate inferences 
from well-specified modules.  More complex examples are considered later.  

\subsection{Cutting feedback implicitly}
We now consider two different ways of defining a cut posterior distribution. The first is through
a Markov chain Monte Carlo (MCMC) sampler,
where some of the full conditional distributions are
modified to remove misspecified
model terms when sampling some of the parameters.
The invariant distribution
of this sampler is an implicitly defined cut posterior. 
To make the ideas easier to describe we first introduce some notation.  

A parametric statistical model is defined for data $\yvec$ with parameters $\thetavec$.  
We consider Bayesian inference with prior density $p(\thetavec)$ and sampling density 
$p(\yvec|\thetavec)$.  Let $\thetavec=(\thetavec_1^\top,\dots, \thetavec_K^\top)^\top$ 
be a partition of $\thetavec$ into $K$ blocks, and assume
the posterior can be factorized as 
\begin{align}
p(\thetavec|\yvec)\propto  p(\thetavec)p(\yvec|\thetavec) & = \prod_{j=1}^M f_j(\thetavec_{A(j)}), \label{fact}
\end{align}
where $A(j)$ denotes the indices of parameter blocks
which appear in $f_j(\cdot)$ (i.e. $k\in A(j)$ if $f_j(\cdot)$
depends on $\thetavec_k$) and $\thetavec_{A(j)}=\left\{\thetavec_k:k\in A(j)\right\}$.  
Factorizations like (\ref{fact}) arise where the joint model
is specified as a directed acyclic graph, but our discussion is more general.  In our notation 
we suppress any dependence of factors $f_j(\cdot)$ on $\yvec$, and the factors
$f_j(\cdot)$ are not uniquely defined.

The \texttt{cut} function in the WinBUGS and OpenBUGS software packages provides a popular
implementation of cutting feedback implicitly; see~\cite{lunn+bsgn09} for a description.  
Write $\thetavec_{-i}=\{\thetavec_j: j\neq i\}$, and consider a Gibbs sampling scheme for the posterior distribution
for $\thetavec$, where we iteratively sample from the full conditional densities
$$p(\thetavec_i|\thetavec_{-i},\yvec)\propto \prod_{j:i\in A(j)}f_j(\thetavec_{A(j)}).$$
On the right-hand side of the above expression we have dropped all terms in the joint model which do not depend on
$\thetavec_i$.  

Suppose we are concerned that one of the factors in the joint model, $f_k(\thetavec_{A(k)})$ say, 
is misspecified, and that it may contaminate inference about some of the other parameters.  
Futhermore, suppose that it is felt that the harmful effects of this misspecification
on inference occur primarily through the influence of this factor on one of the parameter
blocks, without loss of generality $\thetavec_1$ say.  To remedy this we consider 
a modified Gibbs sampler in which sampling from $p(\thetavec_1|\thetavec_{-1},\yvec)$ is replaced with sampling from
$$p_{\text{cut}}(\thetavec_1|\thetavec_{-1},\yvec) \propto \prod_{j:1\in A(j), j\neq k} f_j(\thetavec_{A(j)}).$$
where $f_k(\thetavec_{A(k)})$ has been dropped in forming the full conditional for $\thetavec_1$.
(It is possible to drop multiple factors too).    
The cut posterior distribution is defined here only implicitly through modification of an MCMC algorithm, leading
to a set of possibly inconsistent conditional distributions.
Despite our suggestive notation, $p_{\text{cut}}(\thetavec_1|\thetavec_{-1},\yvec)$
is not the full conditional density of the cut posterior density, but simply denotes the 
conditional density we sample from in the modified MCMC algorithm.
\cite{plummer15} points out that  if we are unable to sample the conditional distributions
exactly, but instead use a Metropolis-within-Gibbs approach, then the distribution defined
through the algorithm depends on the proposal distribution.  

\subsection{Cutting feedback explicitly}\label{sec:explict}

To clarify the cutting feedback approach, \cite{plummer15} considers a simplified two module
system that is nevertheless general enough to cover many situations where cutting feedback is applied in practice.
The two module system is shown in Figure~\ref{two-module}.
\begin{figure}[h]
\centerline{\includegraphics[width=60mm]{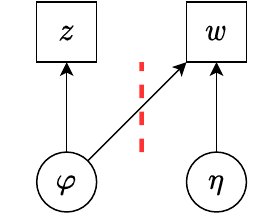}}
\caption{\label{two-module} Graphical representation of a two module system with cutting feedback.  The dashed line indicates
the cut.}
\end{figure}
There are two data sources, which we
denote here as $\wvec$ and $\zvec$.  The likelihood term for $\zvec$ depends on $\varphivec$, and the likelihood
term for $\wvec$ depends on $\etavec$ and $\varphivec$, and $\wvec$ and $\zvec$ are conditionally
independent given $\varphivec$.  We write $\yvec=(\wvec,\zvec)$, and set 
$\thetavec_1=\varphivec$, $\thetavec_2=\etavec$, so that $\thetavec=(\varphivec,\etavec)$.  
The full posterior density of the
joint model can be written in the form (\ref{fact}) as 
$$p(\thetavec|\yvec) \propto p(\thetavec)p(\yvec|\thetavec)=p(\varphivec)p(\etavec|\varphivec)p(\zvec|\varphivec)p(\wvec|\varphivec,\etavec).$$
Figure \ref{two-module} shows the situation where $p(\etavec|\varphivec)$ does
not depend on $\varphivec$, but we allow the more general form of the prior
in what follows. 

The dashed line in Figure \ref{two-module} indicates a ``cut'' in the graph, representing concern that 
$p(\wvec|\varphivec,\etavec)$ may influence inference on
 $\varphivec$. In the modified Gibbs sampling scheme
outlined in the previous subsection, the conditional distribution for $\etavec$ is unchanged, but
the conditional distribution for $\varphivec$ is modified to become
$$p_{\text{cut}}(\varphivec|\etavec,\yvec)\propto p(\varphivec)p(\zvec|\varphivec)\propto p(\varphivec|\zvec).$$
This distribution does not depend on $\etavec$, and the modified Gibbs sampler draws independent samples from the cut joint posterior 
\begin{align}
  p_{\text{cut}}(\varphivec,\etavec|\yvec) & = p(\varphivec|\zvec)p(\etavec|\varphivec,\wvec),  \label{cut-two-module}
\end{align}
where $p(\varphivec|\zvec)$ is the marginal posterior density for $\varphivec$ given $\zvec$, and $p(\etavec|\varphivec,\wvec)$
is the conditional posterior density for $\etavec$ given $\varphivec,\yvec$ (which does not depend on $\zvec$).  

Comparing~\eqref{cut-two-module} with the full joint posterior density
\begin{align}
  p(\varphivec,\etavec|\yvec) & = p(\varphivec|\yvec)p(\etavec|\varphivec,\wvec),  \label{joint-two-module}
\end{align}
where $p(\varphivec|\yvec)$ is the marginal posterior density for $\varphivec$ given $\yvec$, we see that in the cut posterior
density the term $p(\varphivec|\yvec)$ is replaced with $p(\varphivec|\zvec)$.
Thus in the cut posterior, the model for $\wvec$ does not influence inference about
$\varphivec$, while
uncertainty about $\varphivec$ is still propagated when computing inference about $\etavec$.

\subsection{Applications of cutting feedback}

To motivate the importance of the problems we describe it
is helpful to discuss some 
applications of Bayesian modularization and cutting feedback methods.   
The pharmacokinetic and pharmacodynamic models
 mentioned in the introduction were some of the first 
applications, and
\cite{lunn+bsgn09} provides a summary.    
Modularized inference methods 
have also been widely used in causal inference \citep{mccandless+des10,mccandless+rb12,
zigler+wyycd13,pompe+j21}.  Propensity scores are a commonly used
tool in the causal inference literature, and cut methods can be used
to prevent the response model from influencing
the estimation of the propensity scores themselves, while still
propagating uncertainty about them.  
Health effects of air pollution are considered in \cite{blangiardo+hr11}, where
they use survey data to adjust ambient pollution level data in describing
uncertainty in air pollution exposure.  Cut methods can be used here to 
prevent a possibly misspecified module for health outcomes from influencing exposure estimates.     \cite{nicholson-etal21} 
consider the notion of ``interoperability" in modelling for pandemic 
preparedness, which incorporates modularity as one key statistical principle.

Two-stage estimation methods are widely used in econometric analysis, 
where auxiliary models are used to impute observed 
values in the response model. 
These methods are closely related to cutting feedback approaches 
in the Bayesian context.  In two-stage methods, 
accurate propagation of uncertainty in the imputed values is 
important when undertaking inference; see \cite{murphy+t02}. 
A key application is endogeneity correction, which is necessary in many social science studies; for example, in marketing \citep{petrin+t10}.  

\cite{liu+bb09} were motivated to study modularized Bayesian methods by
applications to the analysis of computer models.  
They consider other applications as well, such as meta-analysis, and
this is also considered for a problem in ecology by \cite{ogle+bs13}.  
A semi-modular approach to geographically weighted
regression has been recently discussed in \cite{liu+g21}, where 
choosing how much to pool information spatially in 
a local estimation procedure can be thought of as a problem of managing
model misspecification.  
An interesting archaeological application for cut methods, where 
one of the modules involves only prior terms, is discussed by
\cite{styring17} and \cite{carmona+n20}, and we discuss
this later.  Similar to \cite{styring17}, \cite{moss+r22} also consider
cut methods for priors but in the context of hidden Markov models.

The references we have given here about modularized inference 
applications are not exhaustive, and each of them is often
typical of a larger body of work in a certain discipline. 

\section{Mean field variational inference for cut models}

In this Section we describe variational inference methods for defining cut posterior distributions based on mean
field approximations and variational message passing algorithms.  These methods are analogous to the implicit
definitions of a cut posterior distribution considered in Section 2.1 based on modified Gibbs sampling algorithms.  
There is previous work on variational inference with misspecification \citep{wang+b19} as well as for so-called Gibbs posterior
distributions \citep{alquier+rc16,frazier+lmk21}, 
but variational inference for modularized anlayses are different in the sense that 
a full probabilistic model is assumed, but serious misspecification is confined to only some model components.

\subsection{Variational inference}
Variational approximations of a Bayesian posterior distribution are obtained by minimization of a divergence measure 
between an approximating density and the true posterior density
$p(\thetavec|\yvec)\propto p(\thetavec)p(\yvec|\thetavec)$. The Kullback-Leibler divergence is the most popular metric adopted, so that
for an approximating family ${\cal F}$ of densities,
the optimal approximating density is
\begin{align}
  q^*(\thetavec) & =\arg \min_{q(\thetavec)\in {\cal F}} \text{KL}(q(\thetavec) || p(\thetavec|\yvec)), \label{varopt}
\end{align}
where
\begin{align}
  \text{KL}(q(\thetavec) || p(\thetavec|\yvec)) & = \int \log \frac{q(\thetavec)}{p(\thetavec|\yvec)}q(\thetavec)\,d\thetavec
\end{align}
is the Kullback-Leibler divergence between $q(\thetavec)$ and $p(\thetavec|\yvec)$. It is straightforward to show that 
the optimization at~(\ref{varopt}) is equivalent to
\begin{align}
  q^*(\thetavec) & =\arg \max_{q(\thetavec)\in {\cal F}} {\cal L}(q), \label{varoptmax}
\end{align}
where ${\cal L}(q)$ is called the evidence lower bound (ELBO), and defined as
\begin{align}
  {\cal L}(q) & = \int \log \frac{p(\thetavec)p(\yvec|\thetavec)}{q(\thetavec)} q(\thetavec) d\thetavec.
\end{align}
Overviews of variational inference are given in \citet{ormerod+w10} and \citet{blei+km17}.
The approximating family ${\cal F}$ is usually defined by either a product restriction (leading to 
mean field approximations) or fixed form approximations.  We consider mean field approximations first.

\subsection{Mean field approximations and variational message passing}

Similar to Section 2.1, 
let $\thetavec=(\thetavec_1^\top,\dots, \thetavec_K^\top)^\top$ be a partition of $\thetavec$ into $K$ blocks, and suppose 
${\cal F}$ consists of densities that have the form
\begin{align}
 q(\thetavec) & = \prod_{i=1}^K q_i(\thetavec_i).  \label{mfvb}
\end{align}
Considering the $i$th term $q_i(\thetavec_i)$ in (\ref{mfvb}) and with the terms $q_j(\thetavec_j)$, $j\neq i$ held fixed,   
the value for $q_i(\thetavec_i)$ solving the optimization problem at (\ref{varoptmax}) is
\begin{align}
  q_i^*(\thetavec_i) & \propto \exp\left( E_{-\thetavec_i}(\log p(\thetavec )p(\yvec|\thetavec ) ) \right), \label{coord-ascent}
\end{align}
where $E_{-\thetavec_i}(\cdot)$ denotes expectation with respect to the density $\prod_{j\neq i}q_j(\thetavec_j)$ 
(see, for example, \cite{ormerod+w10}).  
The update (\ref{coord-ascent}) can be used in a coordinate ascent optimization scheme, where after initialization we cycle through the terms 
in $q(\thetavec)$, updating each using (\ref{coord-ascent}) until convergence.  We can also write (\ref{coord-ascent}) as
\begin{align}
  q^*_i(\thetavec_i) & \propto \exp\left( E_{-\thetavec_i}(\log p(\thetavec_i|\thetavec_{-i},\yvec ) ) \right), 
\end{align}
where $p(\thetavec_i|\thetavec_{-i},\yvec)$ is the posterior full conditional distribution for $\thetavec_i$, and 
this formulation shows the close connection between mean field variational inference and Gibbs sampling.  

Consider the factorization of the joint model (\ref{fact}).  Then 
\begin{align*}
\log p(\thetavec)p(\yvec|\thetavec) & = \sum_{j=1}^M \log f_j(\thetavec_{A(j)}),
\end{align*}
and the update (\ref{coord-ascent}) can be written as
\begin{align}
  q^*_i(\thetavec_i) & \propto \prod_{j: i\in A(j)} m_{f_j\rightarrow \thetavec_i}(\thetavec_i), \label{message-passing}
\end{align}
where 
\begin{align*}
  m_{f_j\rightarrow \thetavec_i}(\thetavec_i) & = \exp\left(E_{-\thetavec_i}(\log f_j(\thetavec_{A(j)}))\right).
\end{align*}
The functions $m_{f_j\rightarrow\thetavec_i}(\thetavec_i)$ may be thought of as ``messages" from factor $f_j(\cdot)$ in the model
to $\thetavec_i$.    
For more general discussions of variational message passing algorithms, see \cite{winn+b05}, \cite{minka05}, \cite{knowles+m11} and \cite{wand17}.  \cite{wand17} considers a message passing formulation of mean field approximations involving messages
from both model factors to parameters and factors to parameters, representing the model using a factor graph.  
Computations are formulated in terms of factor graph fragments, and the approach
ensures computational modularity and allows extensions to arbitrarily large models.  

\subsection{Cutting feedback with message passing}

The factorization of the update (\ref{coord-ascent}) into a product of messages motivates one approach to defining
a cut variational posterior distribution.  Write $q^*_1(\thetavec_1),\dots, q^*_K(\thetavec_K)$ for the 
terms of the mean field approximation optimizing the ELBO without cutting feedback.  Similar to 
the discussion of modified Gibbs sampling algorithms in Section 2, suppose that factor
$f_k(\thetavec_{A(k)})$ in the joint model is thought to be misspecified, and that we 
are concerned about the effect of this misspecification on $\thetavec_1$.  We can construct a cut marginal
posterior distribution for $\thetavec_1$ by changing $q^*_1(\thetavec_1)$ to 
\begin{align}
  q_{\text{cut},1}(\thetavec_1) & \propto \prod_{j:1\in A(j), j\neq k} m_{f_j\rightarrow \thetavec_1}(\thetavec_1). \label{cut-marginal}
\end{align}
where the messages in (\ref{cut-marginal}) are the ones obtained at convergence in approximating the true posterior, 
but the message from factor $f_k(\cdot)$ to $\thetavec_1$ is left out.    
After defining a cut marginal posterior distribution for $\thetavec_1$ in this way, we then fix $q_1(\thetavec_1)$ 
to (\ref{cut-marginal}) and optimize the remaining terms 
$q_2(\thetavec_2),\dots, q_K(\thetavec_K)$ using the usual update (\ref{coord-ascent}) until convergence, resulting
in optimal terms $q_{\text{cut},i}(\thetavec_i)$, $i=2,\dots, K$.  The variational cut posterior is then
\begin{align}
  q_{\text{cut}}(\thetavec) & = \prod_{i=1}^K q_{\text{cut},i}(\thetavec_i), \label{cut-joint}
\end{align}
and is an approximation to the joint posterior maximizing the ELBO subject to
constraining the $\thetavec_1$ marginal to be (\ref{cut-marginal}).  
Other modifications of variational message passing algorithms can also be used to produce cut procedures and posteriors.
Algorithm \ref{cut-vmp} describes explicitly 
cut variational message passing for a single cut of the model factor 
$f_k(\cdot)$ on $\thetavec_1$.   
\begin{algorithm}[!h] 
\caption{Cut variational posterior via message passing: 
removing the contribution of $f_k(\cdot)$ to $\theta_1$.}
\label{cut-vmp}
   \vspace{0.1in}
  \noindent {\it Initialization:} 
  
  \begin{itemize}
  \item Initialize $q_2(\thetavec_2),\dots, q_K(\thetavec_K)$.
   \end{itemize}
  
  \noindent{\it Computation of cut posterior density $\prod_{i=1}^K q_{\text{cut},i}(\thetavec_i)$:}
  \begin{enumerate}
  \item Until convergence do:
  \begin{itemize}
  \item For $i=1,\dots, K$: \\
  $$q_i(\thetavec_i)\leftarrow C_i(q)^{-1} \left\{\prod_{j:i\in A(j)} m_{f_j\rightarrow\thetavec_i}(\thetavec_i)\right\},$$
  where $C_i(q)$ is a normalizing constant making $q_i(\thetavec_i)$ integrate to one.  $C_i(q)$ will depend on the factor $i$ being updated and the current value of $q$.
  \end{itemize} 
  \item Calculate
  $$q_{\text{cut},1}(\thetavec_1) = C_{\text{cut},1}^{-1} \prod_{j:1\in A(j),j\neq k}m_{f_j\rightarrow\thetavec_1}(\thetavec_1),$$
  where $C_{\text{cut},1}$ is a normalizing constant making $q_{\text{cut},1}(\thetavec_1)$ integrate to one.
  \item For $i=2,\dots, K$, initialize $q_{\text{cut},i}(\thetavec_i)=q_i(\thetavec_i)$.  
  \item Until convergence do:
  \begin{itemize}
  \item For $i=2,\dots, K$: \\
  $$q_{\text{cut},i}(\thetavec_i) \leftarrow C_i(q_{\text{cut}})^{-1} \left\{\prod_{j: i\in A(j)} m_{f_j\rightarrow \thetavec_i}^{\text{cut}}(\thetavec_i)\right\},$$
  where $C_i(q_{\text{cut}})$ is a normalizing constant making $q_{\text{cut},i}(\thetavec_i)$ integrate to one, and 
  $$m_{f_j\rightarrow \thetavec_i}^{\text{cut}}(\thetavec_i)=\exp
  \left(E_{-\thetavec_i}^{\text{cut}}(\log f_j(\theta_{A(j)}))\right),$$
  where $E_{-\thetavec_i}^{\text{cut}}(\cdot)$ denotes expectation with respect
  to $\prod_{l\neq i} q_{\text{cut},l}(\thetavec_l)$.  
  \end{itemize}
  \item Return $q_{\text{cut}}(\thetavec)=\prod_{i=1}^K q_{\text{cut},i}(\thetavec_i)$.
\end{enumerate}
  
\end{algorithm}

We call the variational message passing approach to evaluating 
the cut posterior ``cut variational message passing", and now elaborate further on
its use in the two module system of Figure \ref{two-module}.  Suppose we
use a factorized variational approximation, 
$$q(\thetavec)=q_{\varphivec}(\varphivec)q_{\etavec}(\etavec).$$
The coordinate ascent update for $q_{\varphivec}(\varphivec)$ is
a product of three messages, 
$$q_\varphivec^*(\varphivec) \propto m_{p(\varphivec)\rightarrow\varphivec}(\varphivec)\times 
m_{p(\zvec|\varphivec)\rightarrow\varphivec}(\varphivec)\times
m_{p(\wvec|\varphivec,\etavec)\rightarrow \varphivec}(\varphivec).$$
It is easy to see from our definition of the messages that 
$m_{p(\varphivec)\rightarrow\varphivec}(\varphivec)=p(\varphivec)$
and $m_{p(\zvec|\varphivec)\rightarrow\varphivec}(\varphivec)=p(\zvec|\varphivec)$, regardless of what $q_\etavec(\etavec)$ is.  Hence when we 
cut, the variational cut marginal for $\varphi$ is 
$$q_{\text{cut},\varphivec}(\varphivec) \propto p(\varphivec)p(\zvec|\varphivec),$$
and hence $q_{\text{cut},\varphivec}(\varphivec)$ is the
exact cut posterior for $\varphivec$.  Once the cut marginal for $\varphivec$ is fixed in this way, 
iteration is not required to find the remaining factor $q_{\text{cut},\etavec}(\etavec)$, which is given by the update \eqref{coord-ascent}.

For non-conjugate models, message passing methods can be difficult to employ
because the messages involve
expectations that cannot be expressed in closed form and are difficult
to compute.
In this case, methods such as 
nonconjugate variational message passing (NCVMP) \citep{knowles+m11} 
and Monte Carlo coordinate ascent variational inference
(MC-CAVI) \citep{ye+bdh19} can be used.  However, we instead use methods based on black box variational inference
and fixed form approximations described in Section 4.

\subsection{Illustrative example revisited}
To illustrate the cut posterior distribution based on variational message passing, we return
to the biased data example in Section~\ref{sec:simpleill}.  
We compare four distributions: the full posterior, cut posterior, and 
variational approximations to both.
The exact cut posterior distribution
has the density at~(\ref{cut-two-module}), which incorporates the cut of the
two module structure depicted in Figure~\ref{two-module}.  Both the full posterior density
and cut posterior density are multivariate normal, and can be computed analytically.

The coordinate ascent updates for $q_\varphi(\varphi)$ and $q_\eta(\eta)$ at~(\ref{message-passing}) are 
\begin{align*}
q^*_\varphi(\varphi) & \propto m_{p(\varphi)\rightarrow \varphi}(\varphi) \times m_{p(\zvec|\varphi)\rightarrow \varphi}(\varphi) \times m_{p(\wvec|\varphi,\eta)\rightarrow \varphi}(\varphi), \\
q^*_\eta(\eta) & \propto m_{p(\eta)\rightarrow \eta}(\eta) \times m_{p(\wvec|\varphi,\eta)\rightarrow \eta}(\eta).
\end{align*}
The expectations required to compute the messages above can be 
evaluated in closed
form, as outlined 
in Appendix~A. 
To specify the variational posterior for $\varphi$, we use the full variational marginal posterior at convergence $q^*_\varphi(\varphi)$
but remove the message $m_{p(\wvec|\varphi,\eta)\rightarrow \varphi}(\varphi)$ to get the variational cut marginal density for $\varphi$, denoted 
$q_{\text{cut},\varphi}(\varphi)$.  
Then we perform a single update for $q_\eta(\eta)$ to obtain
its optimal value with $q_\varphi(\varphi)$ fixed at $q_{\text{cut},\varphi}(\varphi)$ to obtain the cut variational marginal for $\eta$, denoted as $q_{\text{cut},\eta}(\eta)$.  Therefore, 
the final variational cut joint 
posterior is $q_{\text{cut}}(\varphi,\eta)=q_{\text{cut},\varphi}(\varphi)q_{\text{cut},\eta}(\eta)$.   

\begin{figure}[htbp]
\subcaptionbox{Exact (orange) and variational (blue) cut and full posterior densities of $\varphi$. The variational cut and exact cut posterior densities coincide.}[.48\linewidth]{\includegraphics[width=8cm]{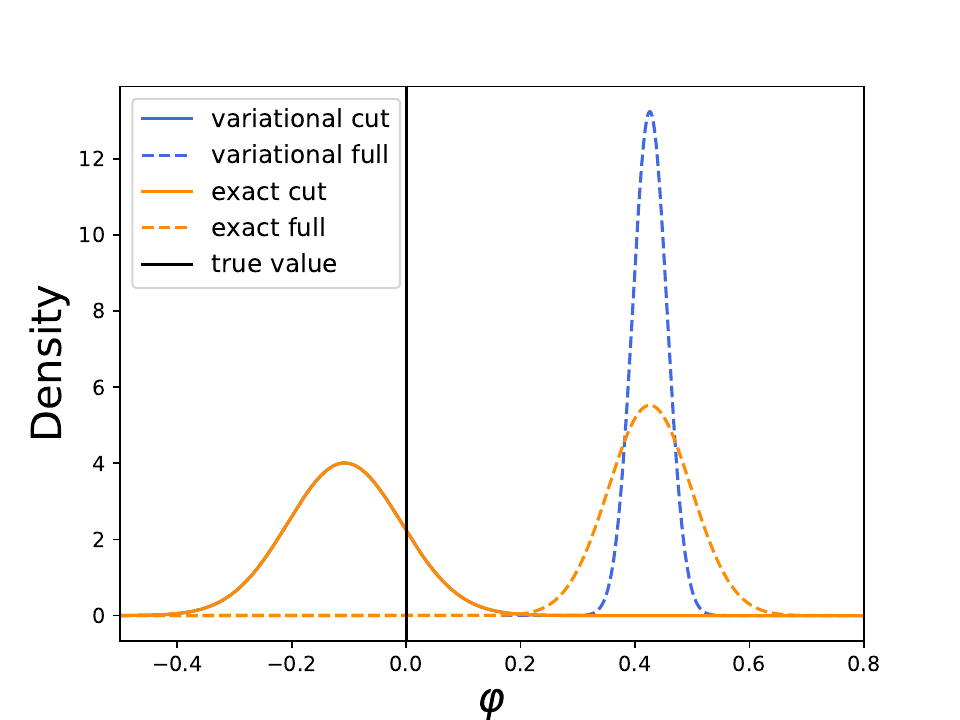}} \hspace{.04\linewidth}
\subcaptionbox{Exact (orange) and variational (blue) cut and full posterior densities of $\eta$.}[.48\linewidth]{\includegraphics[width=8cm]{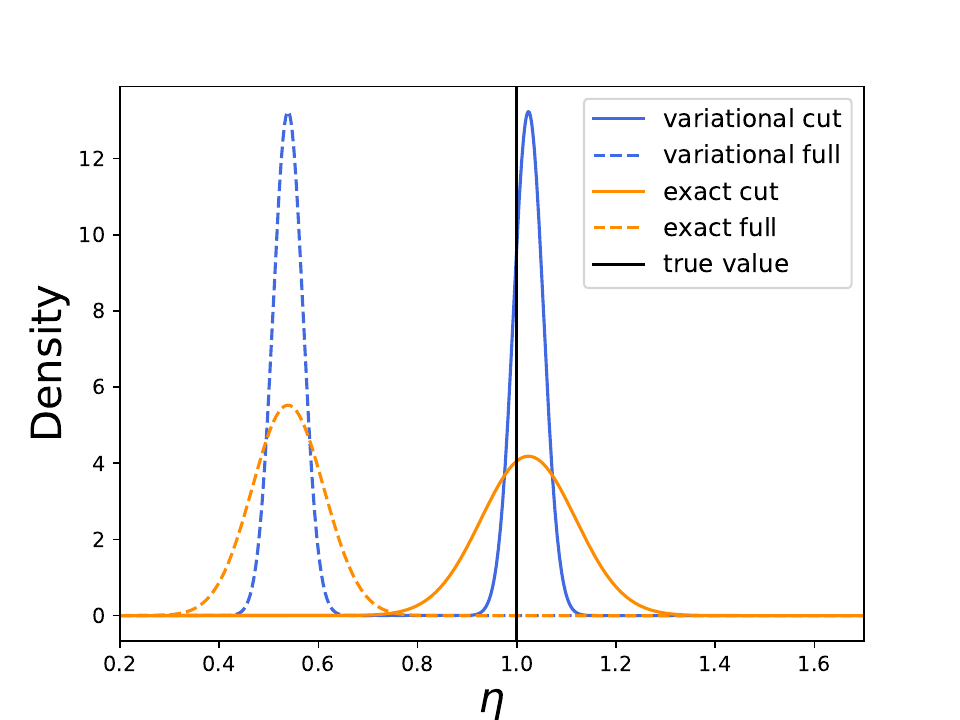}}
\caption{Comparison of marginal posterior estimates of $\varphi$ (left) and $\eta$ (right).}\label{biased-marginal}
\end{figure}

Figure~\ref{biased-marginal} shows the marginal posterior densities for the exact full posterior (dashed orange), the exact 
cut posterior (solid orange), the variational full posterior (dashed blue), the variational cut posterior (solid blue) and the true values (vertical 
black lines).
We make four observations.  First, for parameter $\varphi$ the variational and exact cut posterior are the same in this example, 
which is expected following the discussion in Section 3.3.   
Second, for both $\varphi$ and $\eta$, the full posteriors (whether exact or variational) provide poor
inference, in the sense that the true values $\varphi=0$ and $\eta=1$ lie out in the tails of these distributions.  
The misspecification in the biased data module has contaminated the inference for the full posterior and its variational
approximation.  Third, for both exact and variational cut posteriors, cutting feedback has mitigated
the problem of contamination.  For the cut distributions the true parameter values are in the high probability regions.
Fourth, the variational distributions underestimate uncertainty
compared to their exact counterparts.  

The cause of the underestimation of uncertainty for the variational methods 
is the lack of flexibility of the factorized form of the mean field
approximation.  This is shown in the comparison of the joint posterior distributions for $(\varphi,\eta)$ 
for the exact and variational posterior distributions in
Figure \ref{biased-contour}.  
The variational approximations assume independence between $\varphi$ and $\eta$.  Expressions for the exact posterior density, exact cut posterior density and
variational cut posterior density are given in Appendix A.  The expression 
for the variational cut posterior marginal for $\eta$ shows that it is obtained
as the full conditional density for $\eta$ conditioned on a certain point estimate
for $\varphi$.  This gives an intuitive interpretation of the underestimation
of uncertainty in this marginal density due to ignoring propagation of uncertainty.
The fixed form variational
approximations used in the next section address some of the
problems of mean field approximations 
by allowing greater flexibility for capturing the dependence structure.  

\begin{figure}[h]
	\centerline{\includegraphics[width=100mm]{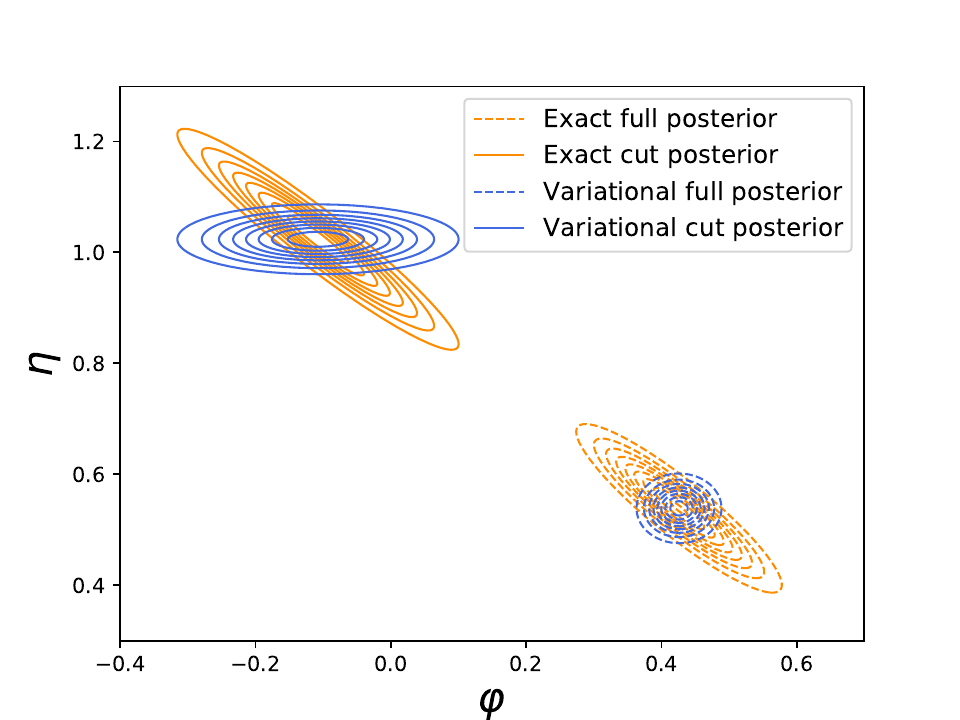}}
	\caption{\label{biased-contour} Contour plots of joint posterior distributions for biased data example. The variational approximations employ a mean field assumption.}
\end{figure}

\section{Defining the cut posterior through optimization}
In this section we consider the two module system depicted in 
Figure~\ref{two-module}. We first highlight why it can be difficult to use
MCMC methods to sample directly from the cut posterior of this 
system at~\eqref{cut-two-module}. However, we then show that the
cut posterior 
can be formulated
as the solution to a constrained optimization problem, which is attractive from an optimization-based 
perspective on Bayesian inference \citep{knoblauch+jd19}. Finally,
motivated by this observation,
we propose fixed form variational approximations that can be used
to compute inference in a computationally attractive fashion. 
In concurrent independent work 
\cite{carmona+n21} have also considered variational inference methods in modularized Bayesian analyses, and they
consider normalizing flows for constructing flexible posterior approximations.   However, the focus of their work
is on semi-modular inference \citep{carmona+n20} rather than cutting feedback.  

\subsection{The cut posterior}
The cut posterior density at~(\ref{cut-two-module}) can be 
expressed as 
\begin{align}
	p_{\text{cut}}(\varphivec,\etavec|\yvec) & = p(\varphivec|\zvec)p(\etavec|\varphivec,\yvec) \nonumber \\
	& \propto p(\varphivec)p(\zvec|\varphivec)\frac{p(\etavec|\varphivec) p(\yvec|\etavec,\varphivec)}{p(\yvec|\varphivec)} \nonumber \\
	& = p(\varphivec)p(\zvec|\varphivec)\frac{p(\etavec|\varphivec) p(\wvec|\etavec,\varphivec)}{p(\wvec|\varphivec)},
	\label{cut-two-module-two}
\end{align}
where we can obtain the last line from the previous one by observing that
$p(\etavec|\varphivec,\yvec)=p(\etavec|\varphivec,\wvec)$ and using the definition of these conditional densities.

Sampling from this posterior distribution using MCMC is difficult, because of the marginal likelihood term
$p(\wvec|\varphivec)$ that appears on the right-hand side of (\ref{cut-two-module-two}), which 
is often intractable. 
Current suggestions in the literature to sample from the cut posterior
are computationally intensive and may not sample exactly from the correct target.  
\citet{plummer15} observes that the cut posterior is often defined only implicitly through 
modification of an MCMC sampling scheme, and 
that the implied posterior density differs according to the proposal used in Metropolis-within-Gibbs schemes.  
The same difficulty has also been pointed out in \citet{woodard+cr13}. 
\citet{plummer15} suggests to use current cut software implementations with caution and performing appropriate
sensitivity analyses. He outlines computationally intensive multiple imputation \citep{little92} and tempered MCMC
approaches to sampling.  Some advanced MCMC methods for
sampling the cut posterior have been considered recently in \citet{jacob+oa17} and \citet{liu+g20}.   \cite{pompe+j21} consider a posterior bootstrap approach
to cut model computation having frequentist validity, and they develop some
asymptotic theory for cut posteriors.  \cite{frazier+n22} study theoretically the
behaviour of cut conditional posterior distributions, which is useful
for understanding how uncertainty propagates between modules.  
Their normal approximations of conditional cut posterior densities are 
useful for both computation and diagnostic purposes.

We now demonstrate that the cut posterior can be formulated
as the solution to a constrained optimization problem, which leads to
variational computational methods. 
Let $q_\varphi(\varphivec)=\int q(\varphivec,\etavec)\mbox{d}\etavec$ be the
marginal density in $\varphivec$ of the approximating density 
$q(\varphivec,\etavec)$. Then define the family of densities 
$${\cal F}_{\text{cut}}=\{q(\varphivec,\etavec):q_\varphi(\varphivec)=p(\varphivec|\zvec)\}.$$
${\cal F}_{\text{cut}}$ is the family of approximations preserving  
the exact cut posterior marginal in $\varphivec$.
The following lemma shows that the density in  ${\cal F}_{\text{cut}}$
that is the best approximation to the full (i.e. uncut) posterior density in
the Kullback-Leibler sense, is given by $p_{\text{cut}}(\etavec,\varphivec|\yvec)$ at~\eqref{cut-two-module}. 
Part~(b) of the lemma states a result which is used further below:  it demonstrates
that the KL divergence between the cut and full posterior is 
the KL divergence between their $\varphivec$ marginals. This result will be useful
later when we develop diagnostic methods for deciding whether or not to cut.\\

\begin{lemma}\label{lem:cutpost} With ${\cal F}_{\text{cut}}$ as defined above, 
\begin{itemize}
\item[(a)] $p_{\text{cut}}(\varphivec,\etavec|\yvec) = \arg \min_{q\in {\cal F}_{\text{cut}}} \text{KL}(q(\varphivec,\etavec) || p(\varphivec,\etavec|\yvec)).$
\item[(b)] $\text{KL}(p(\varphivec,\etavec|\yvec) || p_{\text{cut}}(\varphivec,\etavec|\yvec))=\text{KL}(p(\varphivec|\yvec) || p(\varphivec|\zvec)).$
\end{itemize}
\end{lemma}

\begin{proof}
Write $q(\varphivec,\etavec)\in {\cal F}_{\text{cut}}$ as $q(\varphivec,\etavec)=p(\varphivec|\zvec)q(\etavec|\varphivec)$, 
and $E_q(\cdot)$ for the expectation with respect to $q(\varphivec,\etavec)$.
Since
\begin{align*}
 \log \frac{q(\varphivec,\etavec)}{p(\varphivec,\etavec|\yvec)} & = 
   \log \frac{p(\varphivec|\zvec)q(\etavec|\varphivec)}{p(\varphivec|\yvec)p(\etavec|\varphivec,\wvec)} \\
    & = \log \frac{p(\varphivec|\zvec)}{p(\varphivec|\yvec)}+\log \frac{q(\etavec|\varphivec)}{p(\etavec|\varphivec,\wvec)},
\end{align*}
we have 
\begin{align}
\text{KL}(q(\varphivec,\etavec) || p(\varphivec,\etavec|\yvec)) & = E_q\left( \log \frac{q(\varphivec,\etavec)}{p(\varphivec,\etavec|\yvec)}\right) \nonumber \\
 & = E_q\left(\log \frac{p(\varphivec|\zvec)}{p(\varphivec|\yvec)}\right)
+ E_q\left(\log \frac{q(\etavec|\varphivec)}{p(\etavec|\varphivec,\wvec)}\right).  \label{kldexp}
\end{align}
The first term on the right-hand side of (\ref{kldexp}) does not depend on $q(\varphivec,\etavec)$.  Hence the 
$q(\varphivec,\etavec)$ minimizing the above expression minimizes the second term, 
\begin{align}
 E_q\left(\log \frac{q(\etavec|\varphivec)}{p(\etavec|\varphivec,\wvec)}\right) & = 
  \int \int \log \frac{q(\etavec|\varphivec)}{p(\etavec|\varphivec,\wvec)} q(\etavec|\varphivec) d\etavec\,d\varphivec.  \label{kld-term}
\end{align}
The inner integral on the right-hand side of (\ref{kld-term}) is the KL-divergence between $q(\etavec|\varphivec)$ and
$p(\etavec|\varphivec,\wvec)$.  The integral attains the minimum possible value of zero if we choose, for every $\varphivec$,
$q(\etavec|\varphivec)=p(\etavec|\varphivec,\wvec)$.  This choice corresponds to the cut posterior distribution.  This establishes
part (a) of the lemma. 

Similar to the above argument, we can write $\text{KL}(p(\varphivec,\etavec|\yvec) || p_{\text{cut}}(\varphivec,\etavec|\yvec))$ as the sum of $\text{KL}(p(\varphivec|\yvec) || p(\varphivec|\zvec))$ and the expectation with 
respect to $p(\varphivec|\yvec)$ of $\text{KL}(p(\etavec|\varphivec,\wvec) || p_{\text{cut}}(\etavec|\varphivec,\wvec))$, with
the latter term being zero, which proves (b).
\end{proof}
  
\subsection{Fixed form approximations}
Lemma~\ref{lem:cutpost} motivates
a simple variational approach to approximating the cut posterior.
It employs a variational family of fixed form
approximations with a finite set of parameters 
called ``variational parameters''.

Consider an approximation to $p(\varphivec|\zvec)$ of the form $q_{\widetilde{\lambda}}(\varphivec)$, with parameters 
$\widetilde{\lambdavec}$.  The most accurate approximation in this family (in the Kullback-Leibler sense) 
has parameter values
$$\widetilde{\lambdavec}^*=\arg\max_{\widetilde{\lambdavec}} \int \log \frac{p(\varphivec)p(\zvec|\varphivec)}{q_{\widetilde{\lambda}}(\varphivec)} q_{\widetilde{\lambda}}(\varphivec) d\varphivec.$$  
Solving a variational optimization targeting the exact cut marginal posterior
for $\varphivec$ ensures that the second module data cannot influence the
variational cut marginal posterior density for $\varphivec$.

Writing $\lambdavec=(\widetilde{\lambdavec}^\top,\breve{\lambdavec}^\top)^\top$,
we can then approximate the optimization over ${\cal F}_{\text{cut}}$ in part~(a) of Lemma~\ref{lem:cutpost} by an optimization
over the family 
$$\widetilde{{\cal F}}_{\text{cut}}=\{ q_{\lambda}(\varphivec,\etavec): q_{\lambda}(\varphivec,\etavec)=q_{\widetilde{\lambda}^*}(\varphivec)q_{\breve{\lambda}}(\etavec|\varphivec)\},$$
where the fixed marginal $q_{\widetilde{\lambda}^*}(\varphivec)$ from the first stage optimization approximates
$p(\varphivec|\zvec)$, and we have approximated the conditional posterior distribution for $\etavec$ given $\varphivec$ by
some parametric form $q_{\breve{\lambda}}(\etavec|\varphivec)$ with variational parameters
$\breve{\lambdavec}$.  
The optimization in part~(a) of Lemma~\ref{lem:cutpost} is an optimization targeting the full posterior distribution, so there is no
need to compute the intractable term $p(\wvec|\varphivec)$ in (\ref{cut-two-module-two}).  It is an ordinary variational optimization, where the approximation
to the cut posterior distribution arises from the choice of the variational family
and not through changing the target for the approximation.  
Finally, defining
$$\breve{\lambdavec}^*=\arg\max_{\breve{\lambdavec}} \int \log \frac{p(\varphivec)p(\etavec|\varphivec)p(\zvec|\varphivec)p(\wvec|\varphivec,\etavec)}{q_{\widetilde{\lambda}^*}(\varphivec)q_{\breve{\lambda}}(\etavec|\varphivec)} 
q_{\widetilde{\lambda}^*}(\varphivec)q_{\breve{\lambda}}(\etavec|\varphivec) \,d\varphivec d\etavec,$$
gives an approximation $q_{\lambda^*}(\varphivec,\etavec)=q_{\widetilde{\lambda}^*}(\varphivec)q_{\breve{\lambda}^*}(\etavec|\varphivec)$ to~\eqref{cut-two-module} with calibrated
variational parameters $\lambdavec^*=((\widetilde{\lambdavec}^*)^\top,(\breve{\lambdavec}^*)^\top)^\top$.   $\breve{\lambdavec}^*$ has the interpretation of minimizing the KL divergence of the approximation
to the full posterior distribution within the family 
$\widetilde{{\cal F}}_{\text{cut}}$.

There is another way of viewing the two-stage optimization procedure above, 
which was recently discussed in \cite{carmona+n21}.  
The authors show that
the minimum KL divergence between a given density $q_{\widetilde{\lambda}}(\varphivec)q_{\breve{\lambda}}(\etavec|\varphivec)$ and the densities
in ${\cal F}_{\text{cut}}$ is $\text{KL}(q_{\widetilde{\lambda}}(\varphivec)||p(\varphivec | \zvec))$, 
with this not depending on $\breve{\lambdavec}$.   
This formulates the first stage of the optimization in our approach in a 
similar way to the second, as an optimization over a family of joint
densities for $(\varphivec,\etavec)$.  This allows the entire two stage procedure 
to be seen as optimizing a well-defined objective function 
(\cite{carmona+n21}, Proposition 9) 
by considering a certain family of variational objectives indexed by a hyperparameter, and taking the hyperparameter to a limit.  The results
of \cite{carmona+n21} are in the context of semi-modular inference, 
with cutting feedback as a special case.

While there is a wide range of fixed form densities that can be used
for the variational family, a popular choice for continuous-valued $\thetavec=(\varphivec^\top,\etavec^\top)^\top$ is to assume
\begin{align*}
  q_\lambda(\varphivec,\etavec) & = q_{\widetilde{\lambda}}(\varphivec)q_{\breve{\lambda}}(\etavec|\varphivec)
\end{align*}
is the density of a $N(\muvec,\bm{\Sigma})$ distribution.
If some parameters are constrained, then they can be
transformed to the real line so that they have the same support as the
approximation.
Set $\bm{\Sigma}=\bm{C}\bm{C}^\top$, where $\bm{C}$ is the lower-triangular Cholesky factor,
and partition $\muvec$ and $\bm{C}$ to conform with $\thetavec$, so that 
$\muvec=(\muvec_\varphi^\top,\muvec_\eta^\top)^\top$, 
and
$$\bm{C}=\left[\begin{array}{cc}
  \bm{C}_{\varphi} & \bm{0} \\
  \bm{C}_{\varphi \eta} & \bm{C}_{\eta} 
  \end{array}\right],$$
where $\bm{C}_{\varphi}$ and $\bm{C}_{\eta}$ are both lower triangular.
Here
$\widetilde{\lambdavec}=(\muvec_\varphi^\top,\text{vech}(\bm{C}_{\varphi})^\top)^\top$, where $\text{vech}$ is
the half-vectorization operator which stacks the lower-triangular elements of a matrix column-by-column, and  
$\breve{\lambdavec}=(\muvec_\eta^\top,\text{vec}(\bm{C}_{\varphi \eta})^\top,\text{vech}(\bm{C}_{\eta})^\top)^\top,$
where $\text{vec}$ is the vectorization operator.
Optimization of a Gaussian variational approximation parametrized by its mean vector and lower-triangular Cholesky factor of
its covariance matrix via stochastic gradient methods is considered
by \cite{titsias+l14} and \cite{kucukelbir17} among others.  For the second-stage optimization it is 
straightforward to simply fix the variational parameters $\widetilde{\lambdavec}=\widetilde{\lambdavec}^*$ and 
only optimize over $\breve{\lambdavec}$. 
We do not discuss implementation details
of methods for lower bound gradient estimation for Gaussian approximations, as descriptions 
of this can be found elsewhere, such as in the references given above.  

It is straightfoward to adopt more flexible fixed form approximations for $q_\lambda$ within this framework. 
A simple choice is a Gaussian copula, where a
Gaussian approximation is enriched through learnable marginal transformations~\citep{han+ldc16,smith+ln19}.  
More elaborate variational families can be considered, such as those based on 
mixtures of exponential families (\cite{salimans+k13}, \citet{wu+ks19} among others) or normalizing
flows \citep{papamakarios19}.  
Another direction for obtaining a more flexible approximation 
is to combine variational inference methods with MCMC, and this is considered in Section~\ref{eg:agri} using a method described by \cite{loaiza-maya+snd20}.  In that example
some of the unknowns are discrete, so methods based on continuous 
approximating families do not suffice.  

While diagnosing the adequacy of a particular approximating family can
be challenging, there are some diagnostics that can help.  
\cite{yao+vsg18} consider diagnostics based on Pareto-smoothed
importance sampling corrections and quantile-based simulation based 
calibration, where variational approximations are computed repeatedly
for simulated data.  A moment-based alternative to quantile-based 
calibration is considered in \cite{yu+ntk21}.  The adequacy of an approximation
depends on the use to be made of it, which should inform the way accuracy
is measured.  


\section{Model checks for cutting}

\subsection{Conflict checks}
While both defining and approximating the cut posterior can be challenging, 
another problem in practice is to decide whether or not
to cut.  
We consider now variational implementations of Bayesian model checks that guide this decision. 
Discussion of Bayesian model checking generally can be found
in \cite{gelman+ms96} and \cite{evans15}, while \cite{presanis+osd13}
discuss conflict checking, which may include evaluation of cut posteriors.
\cite{jacob+mhr17} consider a predictive decision-theoretic perspective on 
the decision to cut, and \cite{carmona+n20} consider similar methods 
for their semi-modular inference method and connections with coherent loss-based inference
\citep{bissiri+hw16}.
The conflict-checking approach focuses on the interpretation of the inference
and is complementary to predictive methods. In contrast to 
previous approaches, the checks we propose have two practical advantages. 
The first is that they do not require the specification of non-informative priors in their implementation. The second is that the variational inference
framework can simplify computations greatly.


\subsection{Two module system}
Again, we outline the approach for 
the simplified two module system depicted in 
Figure~\ref{two-module}. The posterior distribution for $\thetavec=(\varphivec^\top,\etavec^\top)^\top$ after observing only $\zvec$ is
\begin{align*}
 p(\varphivec,\etavec|\zvec) & \propto p(\varphivec)p(\etavec|\varphivec)p(\zvec|\varphivec) \\
  & \propto p(\varphivec|\zvec)p(\etavec|\varphivec).
\end{align*}
This can be thought of as a prior that is further updated by subsequent data $\wvec$, so that
\begin{align}
  p(\varphivec,\etavec|\yvec) & \propto p(\varphivec|\zvec)p(\etavec|\varphivec) p(\wvec|\etavec,\varphivec). \label{post-expression}
\end{align}
The $\varphivec$ marginal of (\ref{post-expression}) is $p(\varphivec|\yvec)$.  Lemma 1 (b) shows that the
KL divergence between the cut and full posterior distributions is the KL divergence between $p(\varphivec|\yvec)$ and
$p(\varphivec|\zvec)$.  Hence the KL divergence between the cut and full posterior distributions is a prior-to-posterior
KL divergence for $\varphivec$, when $\zvec$ is known when formulating the prior, and where we consider 
Bayesian updating by the data $\wvec$.  

\cite{nott+wee20} consider conflicts between the prior and posterior using prior-to-posterior KL divergences as a checking
statistic.  Such conflicts occur when the prior puts all its mass in the tails of the likelihood function, so that information
in the prior and data are contradictory.  
We can use such a conflict check to see whether the possibly misspecified
model for $\wvec$ contaminates inference about $\varphivec$.  \cite{nott+wee20} also considered the use of
variational approximations to facilitate implementation of these checks.  

Variational inference methods are
useful for two main reasons here.  First, 
the checks need to be calibrated based on a tail probability
for some reference distribution. Approximating the tail probability involves approximating the posterior distribution
many times for data simulated under the reference distribution; fast variational inference methods
are helpful for this. Second, if the variational approximations to $p(\varphivec|\zvec)$ and 
$p(\varphivec|\yvec)$ are both in exponential family form, then there
is a closed form expression for the KL divergence, which further facilitates computation.  

The method in \cite{nott+wee20} corresponds here to computing the model checking statistic
\begin{align}
 T(\wvec|\zvec) & = \text{KL}\left(p(\varphivec|\yvec) || p(\varphivec|\zvec)\right), \label{klcheck}
\end{align}
which we calibrate through the tail probability
\begin{align}
  p & = \mbox{Pr}\left(T(\mW|\zvec) \geq T(\wvec|\zvec)\right),  \label{tailprob}
\end{align}
where $T(\wvec|\zvec)$ is an observed quantity, while $\mW$ is 
a random vector, 
$$\mW\sim p(\wvec|\zvec)=\int \int p(\wvec|\etavec,\varphivec)p(\varphivec,\etavec|\zvec)d{ \varphivec}\,d\etavec.$$
The tail probability (\ref{tailprob}) gives a measure of how surprising the change is from $p(\varphivec|\zvec)$
to $p(\varphivec|\yvec)$, where the size of the change is calibrated according to what is expected for data $\mW$ simulated
under $p(\wvec|\zvec)$.  By Lemma 1 (b), this check also has the interpretation of how surprising the KL divergence 
between the cut and full posterior distribution is under the same calibration.  Hence the tail probability (\ref{tailprob}) 
is a measure of incompatibility between the inference about $\varphivec$ conditional on $\zvec$, and conditional
on the full data $\yvec=(\wvec,\zvec)$.  

Calculating the tail probability (\ref{tailprob}) is difficult.  Similar to \cite{nott+wee20}, we propose to first replace $T(\wvec|\zvec)$
with
$$\widetilde{T}(\wvec|\zvec)=\text{KL}(q(\varphivec|\yvec) || q(\varphivec|\zvec)),$$
where $q(\varphivec|\yvec)$ and $q(\varphivec|\zvec)$ are variational approximations of $p(\varphivec|\yvec)$
and $p(\varphivec|\zvec)$ respectively.  These approximations will be chosen to have exponential family
forms -- which are popular choices in practice -- for which the KL divergence has a closed form.  
If $q(\varphivec|\yvec)$ and
$q(\varphivec|\zvec)$ are both multivariate normal with means and covariance matrices
$\muvec_{\yvec},\mSigma_{\yvec}$ and $\muvec_{\zvec},\mSigma_{\zvec}$ respectively, then
\begin{align}
\widetilde{T}(\wvec|\zvec) & = \frac{1}{2}\left\{\log \frac{|\mSigma_z|}{|\mSigma_y|}+\mbox{tr}(\mSigma_z^{-1}\mSigma_y)-d+(\muvec_z-\muvec_y)^\top \mSigma_z^{-1}(\muvec_z-\muvec_y)\right\},  \label{normal-KL}
\end{align}
where $d$ is the dimension of $\varphivec$.
Using this approximation to the test statistic (\ref{klcheck}), we draw $S$ samples $\mW^{(i)}$, $i=1,\dots, S$, independently
from $p(\wvec|\zvec)$, and approximate (\ref{tailprob}) by
$$\widetilde{p} = \frac{1}{S} \sum_{i=1}^S \mathbbm{1}(\widetilde{T}(\mW^{(i)}|\zvec) \geq \widetilde{T}(\wvec|\zvec)),$$
where the indicator function $\mathbbm{1}(A)=1$ if $A$ is true, and zero otherwise.  Generation a draw $\mW$ from $p(w|z)$ is done by simulating
$\varphivec, \etavec$ from $p(\varphivec,\etavec|\zvec)$, 
and then for these parameter values 
drawing $\mW$ from $p(\wvec|\varphivec,\etavec)$.  If simulating
from $p(\varphivec,\etavec|\zvec)$ is intractable, we can use its variational
approximation instead.

How well the test statistic $\widetilde{T}(\wvec|\zvec)$ 
approximates $T(\wvec|\zvec)$ is unknown in general, and the tail probability $\widetilde{p}$ may not correspond closely to that obtained from the 
check using $T(\wvec|\zvec)$.  However, this does not really matter.  
The check using the test statistic
$\widetilde{T}(\wvec|\zvec)$ is an ordinary Bayesian model 
check for a valid checking statistic.  
What statistic to use in Bayesian model checking is a free choice of the analyst, 
although it should have a logical motivation.  This is the case here, 
with  $\widetilde{T}(\wvec|\zvec)$ being an informative measure of how
far apart the cut and full posterior distributions are. 

\subsection{Illustrative example revisited}

We revisit the biased data example considered in Sections 2 and 3 to demonstrate our conflict checks.
To perform the check, we first draw $S$ samples $\mW^{(i)}, i = 1,\dots, S,$ independently from $p(\wvec|\zvec)=\int \int p(\wvec|\eta,\varphi)p(\varphi,\eta|\zvec)d\varphi\,d\eta$
as follows.  For $i=1,\ldots,S$, 
\begin{itemize}
	\item Draw  $\varphi^{(i)}$ from $p(\varphi|\zvec)$, which has a normal distribution $\text{N}\left(\frac{n_1 \bar{z}}{n_1 + \delta_1}, \frac{1}{n_1+\delta_1}\right)$, where $\bar{z}$ is the sample mean of $\zvec$.
	\item Draw $\eta^{(i)}$ from the prior $p(\eta)$, which is $\text{N}(0, \delta_2^{-1})$.
	\item Draw $\mW^{(i)} = (w^{(i)}_{1}, \dots, w^{(i)}_{n_2})$ from $p(\wvec|\varphi^{(i)}, \eta^{(i)})$, for which components are independent $\text{N}(\varphi^{(i)} + \eta^{(i)}, 1)$.
\end{itemize}

Write $\mY^{(i)}=(\zvec,\mW^{(i)})$.  
In order to compute the test statistic values $\widetilde T(\wvec|\zvec)$ and $\widetilde{T}(\mW^{(i)}|\zvec)$, $i=1,\dots, S$, we first obtain marginal variational posterior approximations for $\varphi$ conditional on each of $\zvec$, $\mY^{(i)}$ and $\yvec$, denoted as $q(\varphi|\zvec)$, $q(\varphi|\mY^{(i)})$ and $q(\varphi|\yvec)$, respectively. From~(\ref{normal-KL}),
these approximations are normal, so that 
\begin{align*}
	\widetilde{T}(\wvec|\zvec) &=\frac{1}{2} \left\{ \log \frac{\sigma^2_{\zvec}}{\sigma^2_{\yvec}} + \frac{1}{\sigma^2_{\zvec}}(\mu_{\yvec} - \mu_{\zvec})^2 \right\},\\
	\widetilde{T}(\mW^{(i)}|\zvec) &=\frac{1}{2} \left\{ \log \frac{\sigma^2_{\zvec}}{\sigma^2_{\mY^{(i)}}} + \frac{1}{\sigma^2_{\zvec}}(\mu_{\mY^{(i)}} - \mu_{\zvec})^2 \right\}, \quad i=1,\dots, S,
\end{align*}
where $\mu_{\zvec},\mu_{\mY^{(i)}}, \mu_{\yvec}$ and $\sigma^2_{\zvec}, \sigma^2_{\mY^{(i)}}, \sigma^2_{\yvec}$, $i=1,\dots, S$, are the means and variances of the densities $q(\varphi|\zvec), q(\varphi|\yvec_i),$ and $q(\varphi|\yvec_{obs})$.

Figure~\ref{test_statistics} shows the density plot of the test statistics $\widetilde{T}(\mW_i|\zvec), i =1,\dots, S$, with the black 
vertical line being $\widetilde{T}(\wvec|\zvec)$.   We have used a broken $x$-axis in the plot in order to allow
the observed test statistic value and density of the reference distribution to be shown on the same plot.  
The reference density is estimated using $S=100$ simulated datasets.
The observed test statistic is far into the tail of the reference distribution, which
indicates that the difference between cut and full variational posterior densities for the observed data is very large 
compared to what is expected under the prior predictive density for $\wvec$ given $\zvec$.  
This supports a decision to use the cut posterior for inference in this example.
\begin{figure}[h]
	\centerline{\includegraphics[width=90mm]{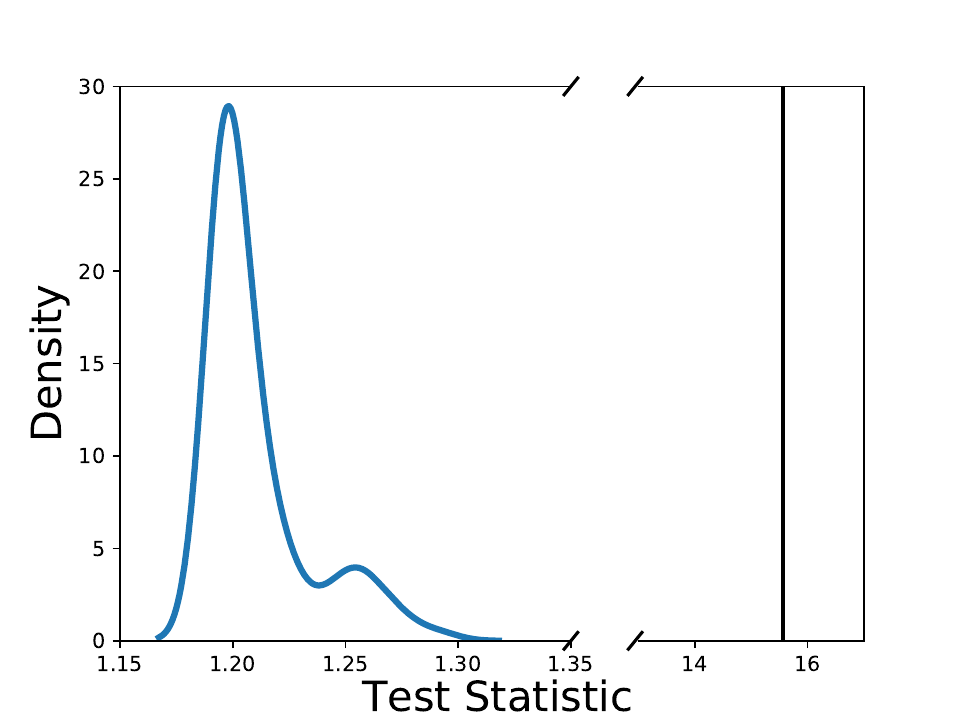}}
	\caption{\label{test_statistics} Observed test statistic value (black vertical line) and estimated reference density
		of the test statistic based on $S=100$ simulations for conflict checking for the biased data example.  A broken $x$-axis is
		used so that the reference density and test statistic can be shown on the same graph.}
\end{figure}

\section{Empirical Examples}
In this section we apply our methodology to two real data analyses
where cutting feedback affects inference substantially.  Code to reproduce
the results can be found at \url{https://github.com/Yu-Xuejun/Variational-Cutting-Feedback}.

\subsection{Human papillomavirus and cervical cancer incidence}
We consider an epidemiological example examined by \cite{plummer15} and motivated
by the study of~\cite{maucort2008international}.  The example considers the relationship between human
papillomavirus (HPV) prevalence and the incidence of cervical cancer, 
and the model used corresponds to a two module system of the type 
described earlier.
In this problem there is data $\zvec=(z_1,\dots, z_{13})^\top$ for
13 countries from an international survey 
where $z_i$ is the number
of people infected with high-risk HPV in a sample of size $n_i$.  There is also data
$\wvec=(w_1,\dots, w_{13})^\top$, for which $w_i$ is the number of cervical cancer cases diagnosed
during $T_i$ years of follow-up.  

The model is
\begin{align*}
&z_i \sim \text{Binomial}(n_i, \gamma_i),\quad \gamma_i \sim \text{Beta}(1,1), \quad \varphi_i=\text{logit}(\gamma_i),\\
& w_i \sim \text{Poisson}(\mu_i),\quad \mu_i = T_i\exp(\eta_1 + \eta_2 \gamma_i), \quad \eta_1,\eta_2 \sim \text{N}(0,10^3).
\end{align*}

Here we are particularly interested in $\eta_2$, which measures the relationship between
HPV prevalence and cancer incidence.  However, there is a concern that the Poisson regression model
is misspecified, and that this may contaminate inferences for both $\gammavec=(\gamma_1,\dots, \gamma_{13})^\top$
and $\etavec=(\eta_1,\eta_2)$.  We write $\varphivec=(\varphi_1,\dots, \varphi_{13})^\top$. 
Although the parameter $\eta_2$ is in the misspecified module,
having a useful interpretation for it crucially depends on the parameter $\gamma_i$ from
the correctly specified module having its intended interpretation.  

First, consider a prior-data conflict check to help decide whether it is useful to cut.
In Figure \ref{hpv-check}, the blue curve is a kernel density estimate obtained from 
simulated test statistic values  $\widetilde{T}(\mW^{(i)}|\zvec)$, $i =1,\dots, 100$, 
where $\mW^{(i)}$ are approximate
simulations from $p(\wvec|\zvec)$, where the variational
approximation for $p(\varphivec|\zvec)$ has been used.  
The black vertical line is the observed test statistic value $\widetilde{T}(\wvec|\zvec)$.
Fixed form normal approximations are used for the posterior distribution of
$(\varphivec^\top,\etavec^\top)$ in obtaining the test statistics. 
Because our variational approximations are normal,
the KL divergence terms for computing the test statistics can be
evaluated in closed form.  
The observed test statistic is larger than all simulated values of the test statistic, so that the difference
between cut and full posterior distributions is large under the reference distribution.  In Appendix~B, Figure \ref{hpv-check-sim} is similar to
Figure \ref{hpv-check}, except that the data are simulated using
the posterior mean values obtained for the real data.  Because the
data are simulated, there is no misspecification for the Poisson regression
in the second module, and Figure \ref{hpv-check-sim} 
demonstrates that the cut and full posterior
distributions are not surprisingly different under the reference distribution
in this case.
\begin{figure}[ht]
	
	\centering
	{\centering\includegraphics[width=90mm]{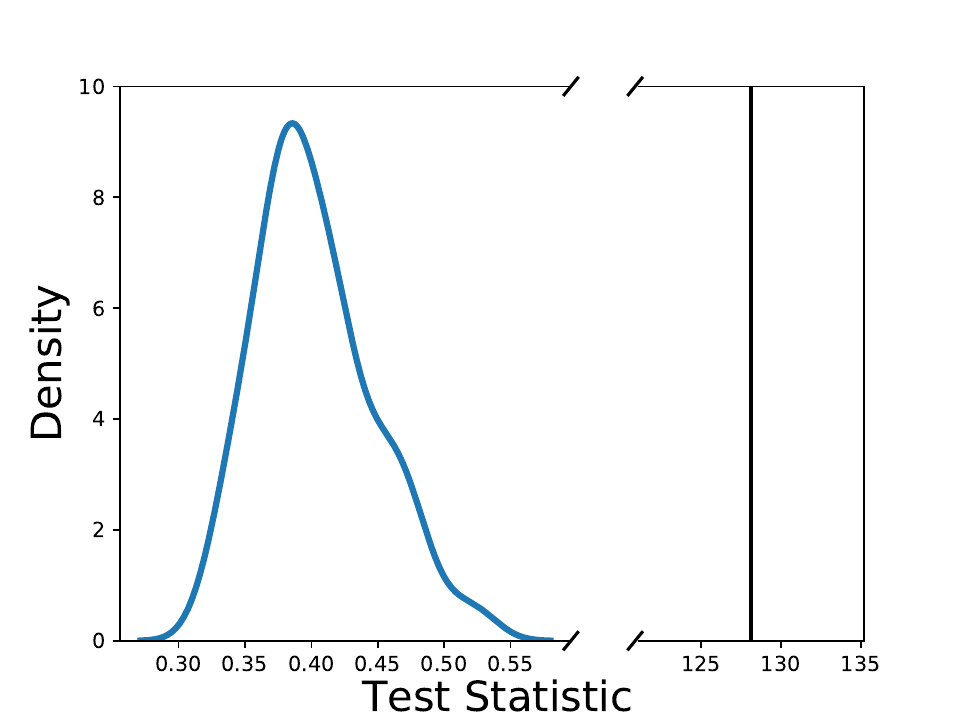}
		
	}
	\caption{Observed test statistic value (black vertical line) and estimated reference density
	for the test statistic based on $100$ simulations for conflict check for the HPV example.  A broken $x$-axis is
	used so that the reference density and test statistic can be shown on the same graph.}
	\label{hpv-check}
\end{figure}

In Figure \ref{hpv-joint}, the exact full and cut posterior distributions are compared with the variational full and cut
posterior distribution approximations.
Draws from the full (i.e. uncut) posterior are shown in blue, and
draws from the
cut posterior in orange.
The left panel 
shows the variational posteriors, while the 
right panel shows the MCMC (i.e. exact) posteriors.
We make two observations. First, the full and cut posterior distributions are quite different, whether computed using MCMC or 
variational inference methods. 
This is consistent with the result of the conflict check.   
Second, the variational approximations are similar to the MCMC estimates of the 
cut and full posterior distributions, although for the full posterior distribution there is a larger difference between
the two approaches.  

It takes 14 minutes to run two-stage Monte Carlo sampling for cut posterior
computation following the multiple
imputation approach of \cite{plummer15}.   
Samples of $\varphivec$ given $\zvec$ are drawn directly from
the exact posterior distribution in a first stage.  This can be done 
because the beta priors
on the probabilities $\gamma_i$ are conjugate to the binomial likelihood terms. 
In the second
stage, for each first stage sample for $\varphivec$, 
a single chain of length $2,000$ is run with the last value drawn 
taken as an approximate sample from the conditional posterior
for $\etavec$ given the fixed value of $\varphivec$.   
The MCMC runs are done using Stan \citep{carpenter-etal-2017}.  
For variational approximation, the computation time is 
68 seconds, using 10,000 stochastic gradient iterations in stage one, 
and 100,000 stochastic gradient iterations in stage two.
Computations were done on an Intel i7-11800H CPU with 8 cores.

\begin{figure}[ht]
	
	\centering
	{\centering\includegraphics[scale=0.5]{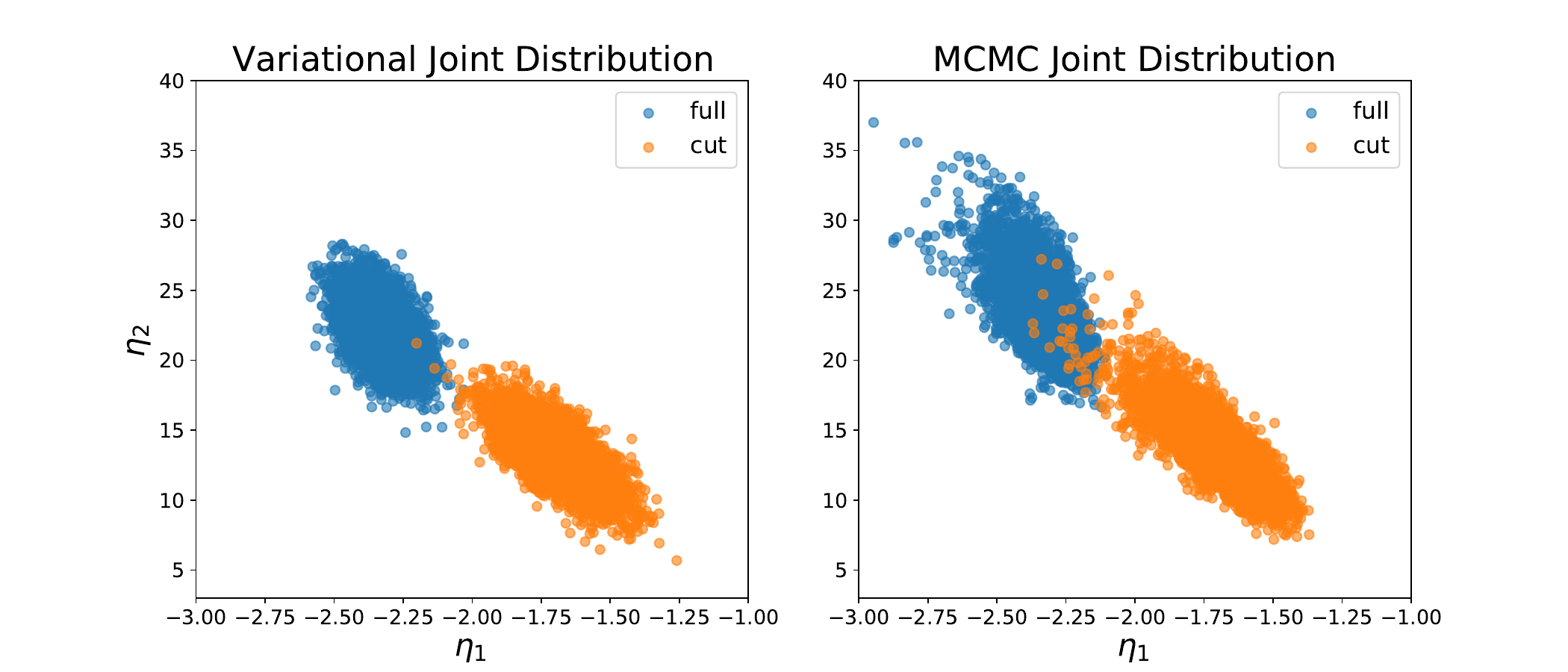}
		
	}
	\caption{Joint full and cut posteriors of $\eta= (\eta_1,\eta_2)$ obtained by fixed form variational approximation and MCMC.}
	\label{hpv-joint}
\end{figure}

\subsection{Agricultural extensification}\label{eg:agri}

\subsubsection{Two module system}
Our last example considers a two module system described in \cite{styring17} and \cite{carmona+n20}. 
The data in this example consists of two parts. The first contains measurements relating
to agricultural practices and productivity from archaeological sites in Northern Mesopotamia, and the second part contains similar modern data obtained under controlled experimental conditions.  
An imputation model is used to account for missingness in the archaeological data.  One of the parameters in the imputation model, which relates site size to manuring levels, 
is of primary interest.  The interpretation of this parameter can 
provide evidence for an ``extensification hypothesis'' of larger land areas being cultivated with
lower manure/midden inputs to support growing urban populations.  
However, the imputation model is rather crude, and it is desirable to cut feedback to ensure
that the interpretation of the key parameter is not influenced by this inadequacy.

\begin{figure}[h]
	\centerline{\includegraphics[width=160mm]{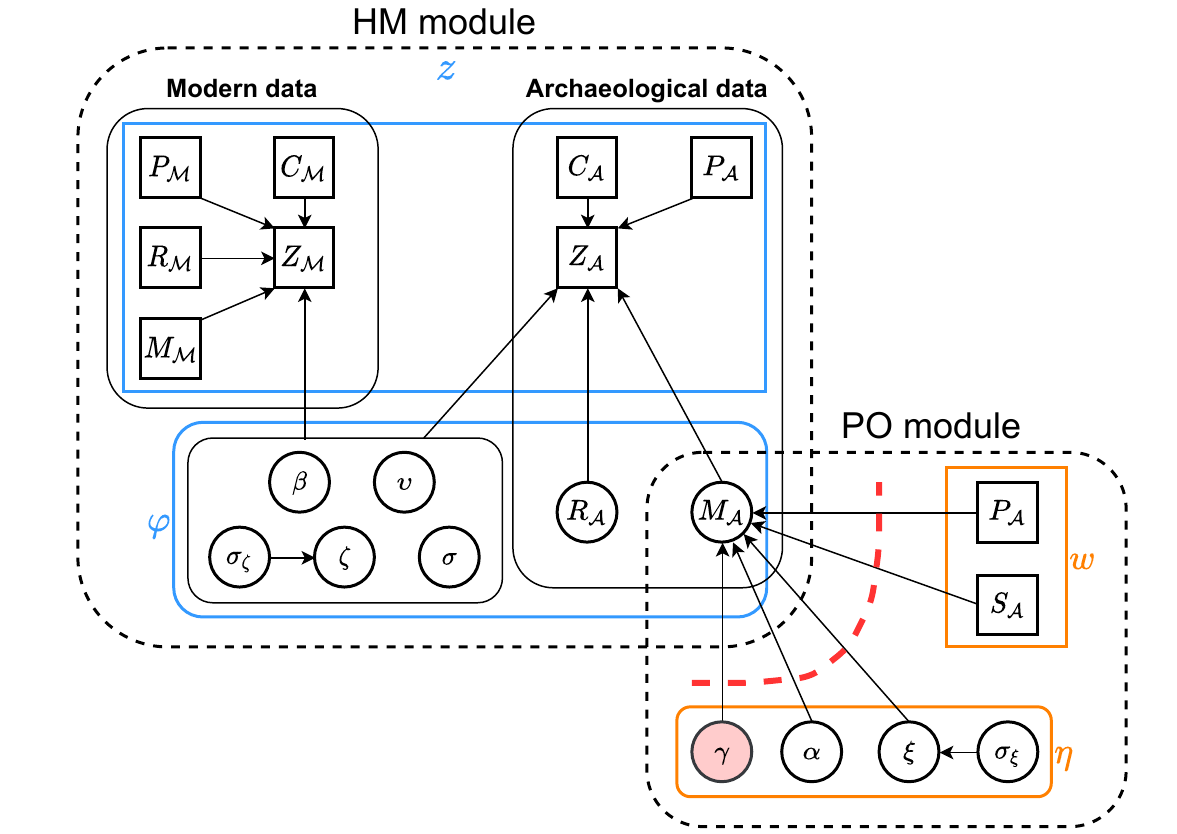}}
	\caption{\label{agri_fig_1} Graphical representation of the two module model for agricultural extensification example. 
		The red dashed line indicates
		the cut. Squares are known data, while circles are unknown parameters or missing data. The pink circle ($\gamma$) denotes the parameter of interest.}
\end{figure}

The archaeological dataset contains variables Nitrogen Level ($Z_{\mathcal{A}}$), Crop Category ($C_{\mathcal{A}}$), Site Location ($P_{\mathcal{A}}$) and Site Size ($S_{\mathcal{A}}$). The modern dataset contains the same four variables denoted 
with subscript $\mathcal{M}$, along with Rainfall ($R_{\mathcal{M}}$) and Manure Level ($M_{\mathcal{M}}$). The latter is an ordinal variable with three possible values $m_{low}< m_{med} < m_{high}$.

A model with two modules is considered here. 
The first module (labeled the ``HM module'' here) is a Gaussian linear regression that pools both datasets and has dependent variable Nitrogen Level. There are fixed
effects in Rainfall and Manuring Level, a random effect in Site Location and 
a different error variance depending on Crop Category.
In this HM module, Rainfall ($R_{\mathcal{A}}$) and Manure level ($M_{\mathcal{A}}$) in the archaeological data are both missing. 
The second module is a proportional odds model (labeled the ``PO module'' here) for imputation of missing values in the archaeological data, with an ordinal response Manure Level. There is
a fixed effect in Site Size with coefficient $\gamma$, a random effect in Site Location and a logit link function. If $\gamma<0$, this provides statistical support for the extensification hypothesis.
Appendix~C details the two modules, along with their likelihood functions
and the priors employed. A graphical depiction of the model, similar to Figure~6 in \cite{carmona+n20}, is shown in Figure~\ref{agri_fig_1}.

\subsubsection{Cutting feedback}
For the full posterior, the PO module plays the role of imputing the missing Manure Level values ($M_{\mathcal{A}}$). However, this model is thought 
to be misspecified, so we cut feedback from the PO module when imputing $M_{\mathcal{A}}$, so that any misspecification does not affect  
interpretation of the parameter of primary interest $\gamma$.  

The notation of a two module system outlined in Section~\ref{sec:explict}
is adopted.  This is depicted in Figure \ref{agri_fig_1}, where for the HM module
the data is denoted as $\zvec$ and unknowns as $\varphivec$,
while for the PO module the site size and location covariates are denoted as $\wvec$ and parameters
as $\etavec$. 
Using this notation, the two module system 
can be further represented in a simplified form in Figure~\ref{agri_fig_2}. While
this differs slightly from that in Figure~\ref{two-module}, 
the cut posteriors for both cases have the same form. By noting that the data 
$\wvec$ consists only of covariate data that is observed without error (i.e. it
is a deterministic quantity),  
the joint posterior is given by
$p(\etavec,\varphivec|\yvec)\propto p(\etavec)p(\varphivec|\etavec,\wvec)p(\zvec|\varphivec)$.   
After cutting feedback, the cut posterior is given at~\eqref{cut-two-module}.  

\begin{figure}[h]
	\centerline{\includegraphics[width=60mm]{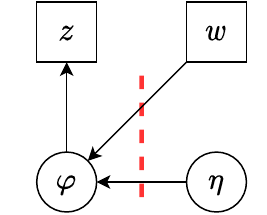}}
	\caption{\label{agri_fig_2} Simplified graphical representation of the two module agriculture model. The red dashed line indicates the cut.}
\end{figure}

Because $\wvec$ is deterministic, it is not possible to perform
the conflict check described in Section 5, because it requires the conditional distribution 
$p(\wvec|\zvec)$ to give a reference distribution. Therefore, we instead measure
conflict for the missing data $M_{\mathcal{A}}$ in the following way. 
Denote by $M_{\mathcal{A},i}$ the $i$th component
of $M_{\mathcal{A}}$, $i=1,\dots, n_{\mathcal{A}}$.  Denote by $p_{\text{cut}}(M_{\mathcal{A},i}=m|\yvec)$ 
and $p(M_{\mathcal{A},i}=m|\yvec)$ the probability
that $M_{\mathcal{A},i}=m$ under the cut and full posterior distributions,
respectively.  
Write $q_{\text{cut}}(M_{\mathcal{A},i}=m)$ and $q(M_{\mathcal{A},i}=m)$ for their
respective variational approximations, which are computed as described below.
Figure~\ref{agri_fig_3} gives pairwise scatterplots of the probability mass values
$\left(q_{\text{cut}}(M_{\mathcal{A},i}=m), q(M_{\mathcal{A},i}=m)\right)$, for all observations in the archaeological data $i=1,\dots, n_{\mathcal{A}}$. 
For points not close to the diagonal line, it indicates that the imputation for $M_{\mathcal{A}}$ is
very different under the cut and full posterior distribution. This is the case here, supporting the decision to cut the posterior in this case.  
\begin{figure}[h]
\begin{center}
\begin{tabular}{ccc}
\includegraphics[width=50mm]{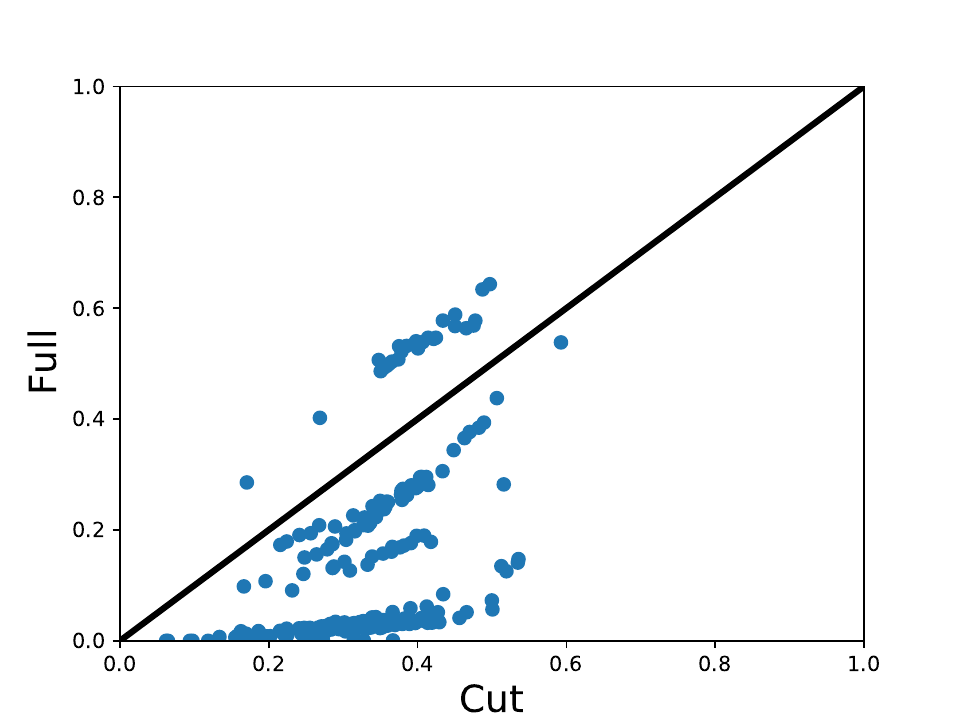} &
\includegraphics[width=50mm]{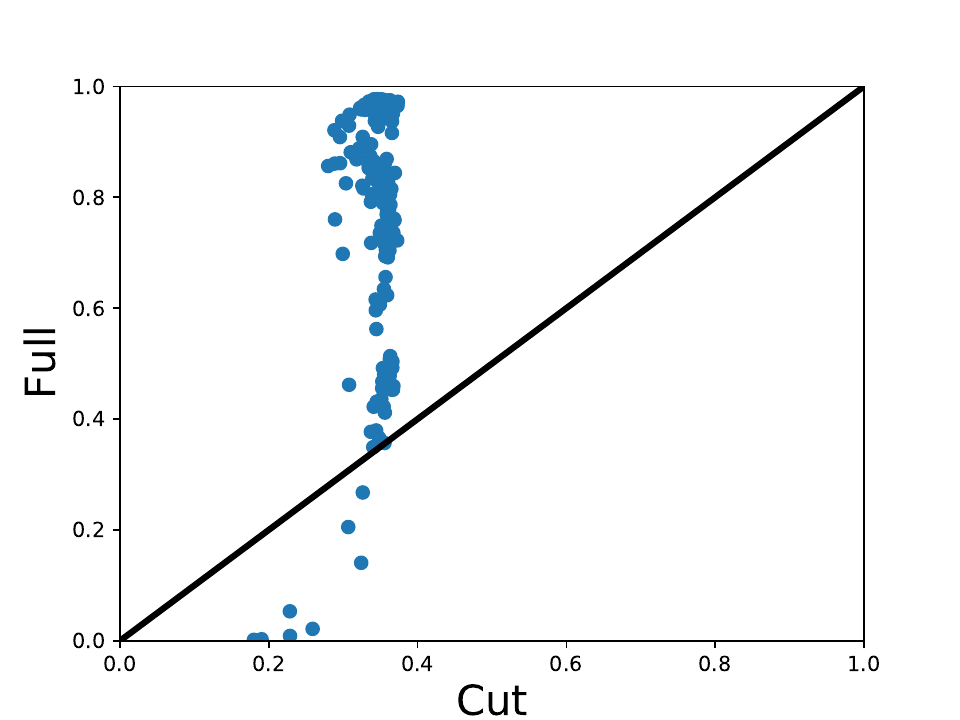} &
\includegraphics[width=50mm]{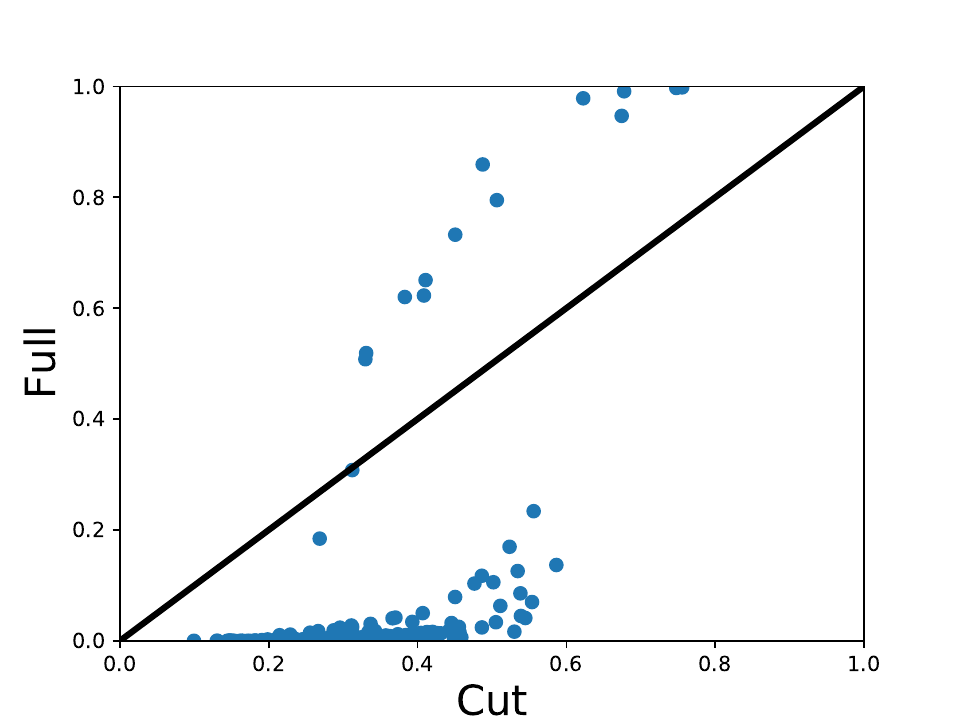}
\end{tabular}
\end{center}
	\caption{\label{agri_fig_3} Depiction of the distribution of the (unobserved)
		manuring level in the archaeological data, estimated using the full and cut variational posteriors. Each panel contains a
		scatterplot of the probability masses $\left(q_{\text{cut}}(M_{\mathcal{A},i}=m), q(M_{\mathcal{A},i}=m)\right)$, for $i=1,\dots, n_{\mathcal{A}}$. From left to right, 
		the panels correspond
		to $m=m_{low}$, $m=m_{med}$ and $m=m_{high}$.}
\end{figure}

\subsubsection{Variational inference}
A two-stage variational optimization is performed to get a Gaussian approximation of the posterior of the continuous parameters for 
the cut model. Parameters in the HM module are updated by the first-stage optimization whereas those in the PO module are updated by the second-stage optimization with $M_{\mathcal{A}}$ fixed to be the variational posterior samples obtained by the first stage.
The difficulty in this example is that, in the HM module, the three-level discrete missing value $M_{\mathcal{A}}$ cannot be approximated by a Gaussian distribution. In this case, we use a  method recently developed by~\cite{loaiza-maya+snd20} which treats $M_{\mathcal{A}}$ as a latent variable and updates it by Monte Carlo generation,
while updating the variational parameters by stochastic optimization as we now discuss.

In the HM module the unknowns are $\varphivec = (\rhovec^\top, M_{\mathcal{A}}^\top)^\top$, where $\rhovec$ denotes all the unknowns other than
$M_{\mathcal{A}}$.
The posterior density $p(\varphivec|\zvec) = p(\rhovec, M_{\mathcal{A}} | \zvec) \propto p(\zvec | \rhovec,M_{\mathcal{A}})p(M_{\mathcal{A}}|\rhovec)p(\rhovec) \equiv g(\rhovec, M_{\mathcal{A}})$ can be approximated by the variational density
\begin{equation}
	q_\lambdavec(\rhovec, M_{\mathcal{A}}) = p(M_{\mathcal{A}}|\rhovec, \zvec) q^0_\lambdavec(\rhovec)\,.\label{eq:lmsnd}
\end{equation}
A Gaussian variational approximation $q^0_\lambdavec(\rhovec)=\phi(\rhovec;\muvec,\mC\mC^\top)$ is adopted for $\rhovec$,
while the exact conditional posterior is used for $M_{\mathcal{A}}$.
\cite{loaiza-maya+snd20} show that~\eqref{eq:lmsnd} is an accurate variational
approximation, and that Algorithm~\ref{agri_alg_1} is a fast 
stochastic gradient ascent algorithm to minimize the variational lower bound. 
At step~(2) of this algorithm, 
$M_{\mathcal{A}}$ is generated from its conditional posterior which is available analytically here.  In step~(3), it is possible to use automatic differentiation
tools for computing the gradient estimate with the discrete samples.  
In step~(4), the
step sizes are obtained adaptively using the ADADELTA method \citep{zeiler12}.  
 We keep the last 10000 updates of $\varphivec = (\rhovec, M_{\mathcal{A}})$ as
 a sample from the variational posterior, and randomly draw from it in each update in the second-stage optimization for the PO module, so as to account for posterior uncertainty in the HM module.

\begin{algorithm}
	\caption{Variational Inference with Latent Variables $M_{\mathcal{A}}$} 
	\label{agri_alg_1}
	\begin{algorithmic}
		\STATE Initialize $\lambdavec = (\muvec, \mC)$, $M_{\mathcal{A}}$, run-length $K$
		\FOR{$k = 1:K$} 
		\STATE (1) Generate $\epsilonvec \sim \text{N}(\bm{0},I)$ and set $\rhovec = \mC \epsilonvec + \muvec$
		\STATE (2) Generate $M_{\mathcal{A}} \sim p(M_{\mathcal{A}} |\rhovec, \zvec)$ exactly
		\STATE (3) Compute an unbiased estimate of the lower bound gradient, 
		$$\widehat{\nabla_\lambda {\cal L}(\lambda)} = \frac{d\rhovec}{d\lambdavec}^\top \left\{\nabla_\rhovec \log g(\rhovec, M_{\mathcal{A}})-\nabla_{\rhovec} \log q_\lambdavec^0(\rhovec)\right\}$$
		\STATE (4) Update 
		$$\lambdavec = \lambdavec + \svec_k\circ \widehat{\nabla_\lambda {\cal L}(\lambda)},$$
		where $\svec_k$ is a vector of step sizes at iteration $k$ and $\circ$ denotes elementwise product
		for two vectors.  
		\ENDFOR
	\end{algorithmic}
\end{algorithm}

\begin{figure}[htbp]
\subcaptionbox{Cut posterior of $\gamma$.}[.48\linewidth]{\includegraphics[width=8cm]{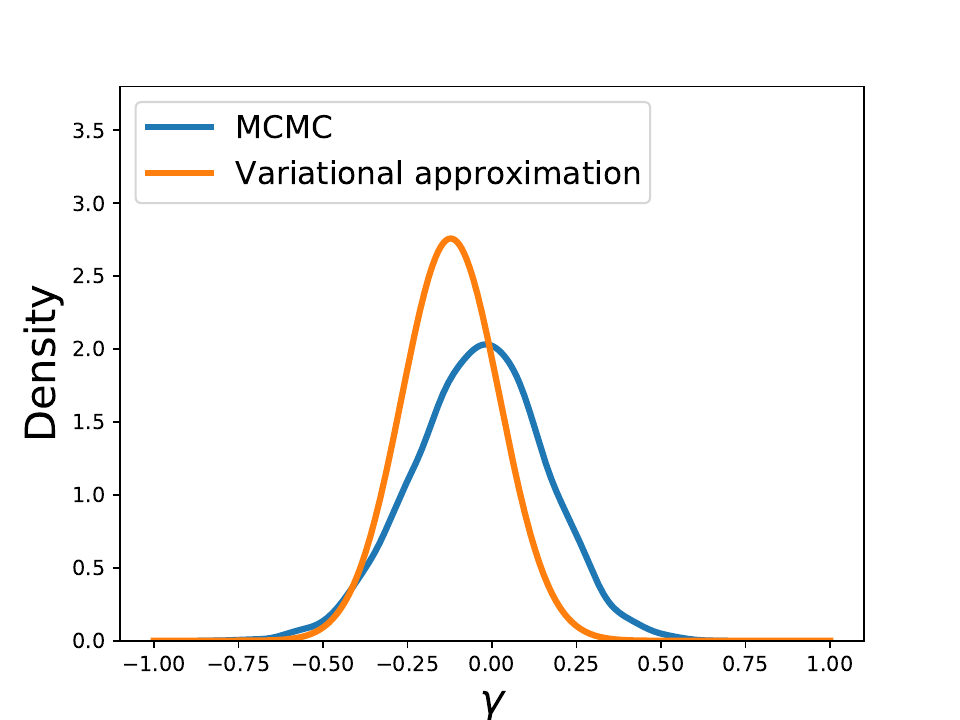}} \hspace{.04\linewidth}
\subcaptionbox{Full posterior of $\gamma$.}[.48\linewidth]{\includegraphics[width=8cm]{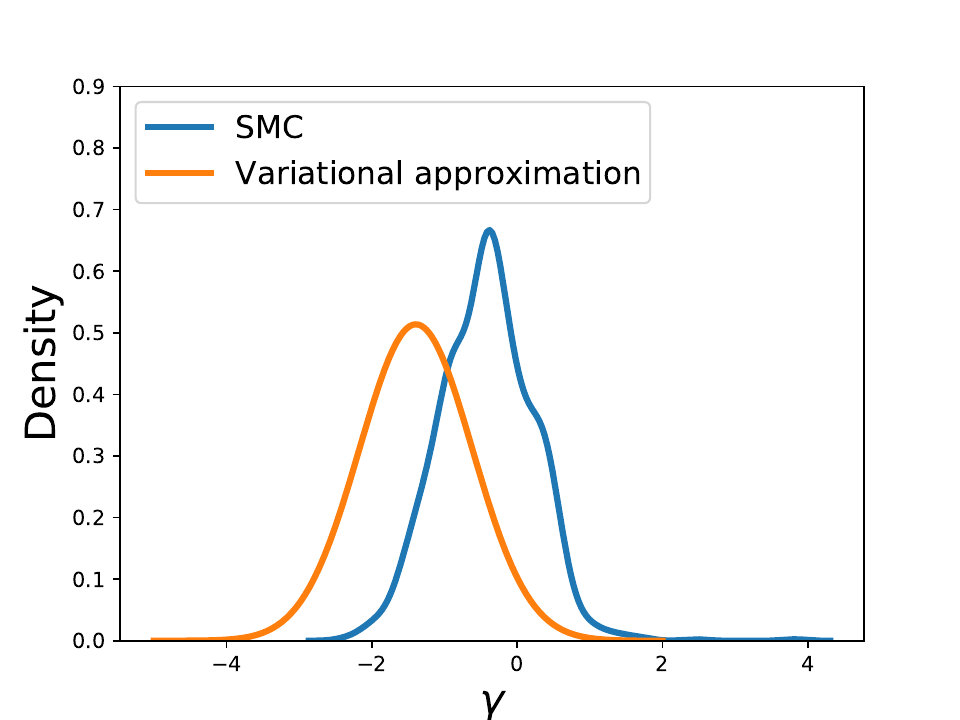}}
\caption{Posterior of $\gamma$ obtained by cut model (left) and full model (right).}\label{gamma}
\end{figure}

Figure \ref{gamma} shows the marginal full and cut posterior densities for $\gamma$ obtained by variational approximation and
exact-in-principle Monte Carlo methods.  For comparison with the variational methods, 
MCMC was used to approximate the cut posterior 
distribution, and a sequential Monte Carlo (SMC) sampler \citep{delmoral+dj06} for the full posterior distribution.   
The full posterior is complex, and we found it necessary to use the SMC sampler instead of MCMC to evaluate
it reliably.
The posterior densities obtained using variational approximation and MCMC are shown by orange and
blue lines, respectively.  We make two observations.  First, the variational posterior distributions for $\gamma$ are similar
to those obtained by MCMC for the cut distribution, with a larger difference between the variational and SMC method
for the full posterior distribution.  Second, the inferences
are very different for the cut and full posterior distributions.  
The inferences we obtain for the cut model are similar to those in \cite{carmona+n20}, although we have
used different priors in our analysis.
In the semi-modular inference approach of \cite{carmona+n20} they
consider partially cutting feedback, with the amount of feedback chosen according to predictive criteria.  
This results in a posterior probability of $\gamma<0$ (supporting the extensification hypothesis) 
somewhere in between those for the cut and full posterior
distributions.

The computation time for cut posterior approximation using the 
two-stage variational 
approach and MCMC are as follows. For a multiple imputation MCMC
approach \citep{plummer15} samples are generated for $\varphivec$
from $p(\varphivec|\zvec)$ in a first stage.  This takes 21 hours 
for MCMC with 100,000 iterations and 2 chains.  We retain
$1,000$ samples from the first stage.  
At a second stage, generating an approximate
conditional full posterior sample of $\etavec$ for each of the
first stage $\varphivec$ samples takes 33 minutes using a single chain
of length $10,000$ and retaining the last value.
In the two-stage variational approach, the first stage takes 93 minutes
in stage one and 3 minutes in stage two for 100,000 stochastic
gradient ascent iterations in each case.  Computations were done 
on an Intel i7-11800H CPU with 8 cores.

The full and cut posterior densities in this example are complex.  
In examples like this one, the reader might wonder what happens 
if the approximating family is poorly chosen.  In this case, if
the conventional KL variational objective is used, then underestimation
of uncertainty is a common result, although variational point estimates
can still perform well.  Many authors
have considered using 
alternative divergence measures in variational inference, 
which may result in approximations which are mass-covering.  
However, for these alternative approaches stable optimization in high dimensions
can be challenging.

\section{Discussion}

Variational methods have strong potential in modularized inference, including the case in which cutting feedback
between modules is required.  
This paper develops cut procedures using variational inference methods which have reduced computational demands
compared to existing MCMC implementations.  We consider both cut procedures defined through modifications of
variational message passing, as well as explicit formulations of cutting feedback where the cut posterior can be defined
as the solution of a constrained optimization problem.  In the explicit formulation, it is convenient to use fixed form variational
approximations based on a sequential decomposition, which also leads to practical variational implementations 
of computationally burdensome checks for conflicting information that are useful in making the decision of whether or not 
to cut.

In the message passing formulation of cutting feedback there are alternative possible implementations that have not
been explored here, and we leave this to future work.  Our work has mostly used simple Gaussian approximations to
facilitate the conflict checks discussed in Section 4, where the ability to explicitly compute KL divergences is an
advantage.  Further interesting work could be done on using more flexible variational families, similar to the last
example of Section 5 where we considered combining MCMC and variational inference methods using
the approach of \cite{loaiza-maya+snd20}.  The recent work in \cite{carmona+n21} using normalizing flows
is another promising direction.  It is interesting to ask whether variational message passing or stochastic gradient cut methods are preferred when both can be implemented.  In the two module 
case, there can be strong dependence between $\varphivec$ and $\etavec$,
and this makes a factorized approximation in variational message
passing unattractive.  
The generality and ease of implementation
of stochastic gradient optimization using automatic differentiation tools
make stochastic gradient cut approximations the preferred
approach in many situations.  However, we do think the variational message
passing approach may have uses in complex situations beyond the two module
case.

\section*{Acknowledgements}

The authors thank Chris Carmona and Geoff Nicholls for sharing some details of their work with us, and the review team for helpful feedback that improved the paper.

\section*{Appendix A: Variational message passing for biased data example}

In the biased data example, data likelihoods and priors can be summarized as follows
\begin{align*}
	\zvec = (z_1, \dots, z_{n_1})|\varphi \sim N(\varphi, 1) \quad i.i.d,\quad  \varphi \sim N(0, \delta_1^{-1}),\\
	\wvec = (w_1, \dots, w_{n_2})|\varphi,\eta \sim N(\varphi+\eta,1) \quad i.i.d, \quad \eta \sim N(0, \delta_2^{-1}),
\end{align*}
where $\delta_1$ and $\delta_2$ are known.

The variational joint posterior at~\eqref{mfvb} has the form
\begin{align*}
	q(\varphi, \eta) = q_\varphi(\varphi)q_\eta(\eta).
\end{align*}
The coordinate ascent updates for $q_\varphi(\varphi)$ and $q_\eta(\eta)$ for approximating the full posterior can be written as
\begin{align*}
	q_\varphi(\varphi) & \propto m_{p(\varphi)\rightarrow \varphi}(\varphi) \times m_{p(\zvec|\varphi)\rightarrow \varphi}(\varphi) \times m_{p(\wvec|\varphi,\eta)\rightarrow \varphi}(\varphi), \\
	q_\eta(\eta) & \propto m_{p(\eta)\rightarrow \eta}(\eta) \times m_{p(\wvec|\varphi,\eta)\rightarrow \eta}(\eta).
\end{align*}
where
\begingroup
\allowdisplaybreaks
\begin{align*}
	m_{p(\varphi)\rightarrow \varphi}(\varphi) & = \exp\left(E_{q_\eta}(\log p(\varphi))\right)=p(\varphi)= \sqrt{\frac{\delta_1}{2\pi}} \exp (-\frac{1}{2} \delta_1 \varphi^2),\\
	m_{p(\eta)\rightarrow \eta}(\eta) & = \exp\left(E_{q_\varphi}(\log p(\eta))\right)=p(\eta)= \sqrt{\frac{\delta_2}{2\pi}} \exp (-\frac{1}{2} \delta_2 \eta^2),\\
	m_{p(\zvec|\varphi)\rightarrow \varphi}(\varphi) & = \exp\left(E_{q_\eta}(\log p(\zvec|\varphi))\right)=p(\zvec|\varphi)= \frac{n_1}{\sqrt{2\pi}} \exp (-\frac{1}{2} \sum_{i=1}^{n_1}(z_i - \varphi)^2),\\
	m_{p(\wvec|\varphi,\eta)\rightarrow \varphi}(\varphi) & = \exp\left(E_{q_\eta}(\log p(\wvec|\varphi,\eta))\right)\\ &= \exp(-\frac{n_2}{2}\log(2\pi) - \frac{1}{2} \sum_{i=1}^{n_2} E_{q_\eta}[(w_i - \varphi - \eta)^2])\\
	&= \exp(-\frac{n_2}{2}\log(2\pi) - \frac{1}{2} \sum_{i=1}^{n_2}[Var_{q_\eta}(\eta)+ (w_i-\varphi- E_{q_\eta}(\eta))^2]),\\
	m_{p(\wvec|\varphi,\eta)\rightarrow \eta}(\eta) & =\exp\left(E_{q_\varphi}(\log p(\wvec|\varphi,\eta))\right)\\
	&= \exp(-\frac{n_2}{2}\log(2\pi) - \frac{1}{2} \sum_{i=1}^{n_2} E_{q_\varphi}[(w_i - \varphi - \eta)^2])\\
	&= \exp(-\frac{n_2}{2}\log(2\pi) - \frac{1}{2} \sum_{i=1}^{n_2}[Var_{q_\varphi}(\varphi)+ (w_i-E_{q_\varphi}(\varphi)- \eta)^2]).\\
\end{align*}
\endgroup

Thus, 
\begingroup
\allowdisplaybreaks
\begin{align*}
	q^*(\varphi) & \propto m_{p(\varphi)\rightarrow \varphi}(\varphi) \times m_{p(\zvec|\varphi)\rightarrow \varphi}(\varphi) \times m_{p(\wvec|\varphi,\eta)\rightarrow \varphi}(\varphi) \\
	&\propto \exp(-\frac{\delta_1}{2} \varphi^2) \exp(-\frac{1}{2}\sum_{i=1}^{n_1}(z_i -\varphi)^2) \exp(- \frac{1}{2} \sum_{i=1}^{n_2}(w_i-\varphi- E_{q_\eta}(\eta))^2)\\
	&\propto \exp(-\frac{1}{2}[(\delta_1 + n_1+n_2)\varphi^2 - 2(\sum_{i=1}^{n_1}z_i + \sum_{i=1}^{n_2}(w_i-E_{q_\eta}(\eta))\varphi]),\\
	q^*(\eta) & \propto m_{p(\eta)\rightarrow \eta}(\eta) \times m_{p(\wvec|\varphi,\eta)\rightarrow \eta}(\eta)\\
	&\propto \exp(-\frac{\delta_2}{2} \eta^2) \exp(-\frac{1}{2}\sum_{i=1}^{n_2}(w_i - E_{q_\varphi}(\varphi) - \eta)^2)\\
	&\propto \exp(-\frac{1}{2}[(\delta_2+n_2)\eta^2 -2(\sum_{i=1}^{n_2}(w_i-E_{q_\varphi}(\varphi)))\eta]).
\end{align*}
\endgroup
The variational posterior densities of $\varphi$ and $\eta$ are both normal distributions. To update them, we only need to update means and variances by
\begin{align*}
	&E_{q_\varphi}(\varphi) = \frac{n_1 \bar{\zvec} + n_2 \bar{\wvec} - n_2 E_{q_\eta}(\eta)}{\delta_1 + n_1 +n_2}, \quad \quad Var_{q_\varphi}(\varphi) = \frac{1}{\delta_1 + n_1 + n_2},\\
	&E_{q_\eta}(\eta) = \frac{n_2\bar{\wvec} - n_2 E_{q_\varphi}(\varphi)}{\delta_2 + n_2}, \quad \quad
	Var_{q_\eta}(\eta) = \frac{1}{\delta_2 + n_2}.
\end{align*}

The cut marginal posterior density for $\varphi$ is just $p(\varphi|z)$ by
the discussion in Section 3.3.  Hence after some simple algebra
$$q_{\text{cut},\varphi}(\varphi)=p(\varphi|z)=\phi\left(\varphi;(n_1+\delta_1)^{-1}n_1\bar{z},(n_1+\delta_1)^{-1}\right).$$

To get $q_{\text{cut},\eta}(\eta)$, we plug the posterior
mean of $\varphi$ from the cut posterior density into $q^*(\eta)$ to obtain
$$q_{\text{cut},\eta}(\eta) = \phi\left(\eta;(n_2+\delta_2)^{-1} n_2(\bar{w}-(n_1+\delta_1)^{-1} n_1\bar{z}),(n_2+\delta_2)^{-1}\right).$$

The exact joint cut posterior density is $p(\varphi|z)p(\eta|\varphi,w)$, 
where $p(\varphi|z)$ is given above, and 
\begin{align*}
  p(\eta| \varphi,w) & \propto p(\eta)p(w|\varphi,\eta),
\end{align*}
which simplifies after some algebra to 
$$p(\eta|\varphi,w)=\phi\left(\eta; (\delta_2+n_2)^{-1} n_w (\bar{w}-\varphi),(\delta+n_2)^{-1}\right).$$
We can see that $q_{\text{cut},\eta}(\eta)=p(\eta|\varphi=\mu_\varphi,w)$, where
$\mu_\varphi=n_1\bar{z}/(n_1+\delta_1)$.  So the cut marginal variational posterior density
for $\eta$ is obtained by plugging in a point estimate of $\varphi$ to the posterior
full conditional for $\eta$.

Finally, the exact Bayes posterior is
$$p(\varphi|z,w)p(\eta|\varphi,w),$$
with $p(\eta|\varphi,w)$ given above,and 
$$p(\varphi|z,w)=\phi\left(\left( n_1+\delta_1+\frac{n_2\delta_2}{\delta_2+n_2}\right)^{-1} \left( n_1\bar{z}+\frac{n_2\delta_2}{\delta_2+n_2}\bar{w} \right), \left( n_1+\delta_1+\frac{n_2\delta_2}{\delta_2+n_2}\right)^{-1} \right).$$
It can be seen from this expression that if $n_2$ and $\delta_2$ are large, then
the posterior mean for $\varphi$ will be dominated by the sample mean
obtained from the biased data.

\section*{Appendix B: Conflict checking for simulated data, HPV example}

For the posterior mean parameter values obtained via MCMC for the HPV
example of Section 6.1, we simulated a replicate dataset.  Since the
data are simulated, there is no misspecification of the likelihood in the
second module.   We then repeated the conflict check shown in Figure
\ref{hpv-check}, for the simulated data.  The result is shown in Figure
\ref{hpv-check-sim}.  As might be expected, the difference between
the full and cut variational posterior distribution is not large compared
to the reference distribution in this setting of correct model specification.
\begin{figure}[ht]
	
	\centering
	{\centering\includegraphics[width=90mm]{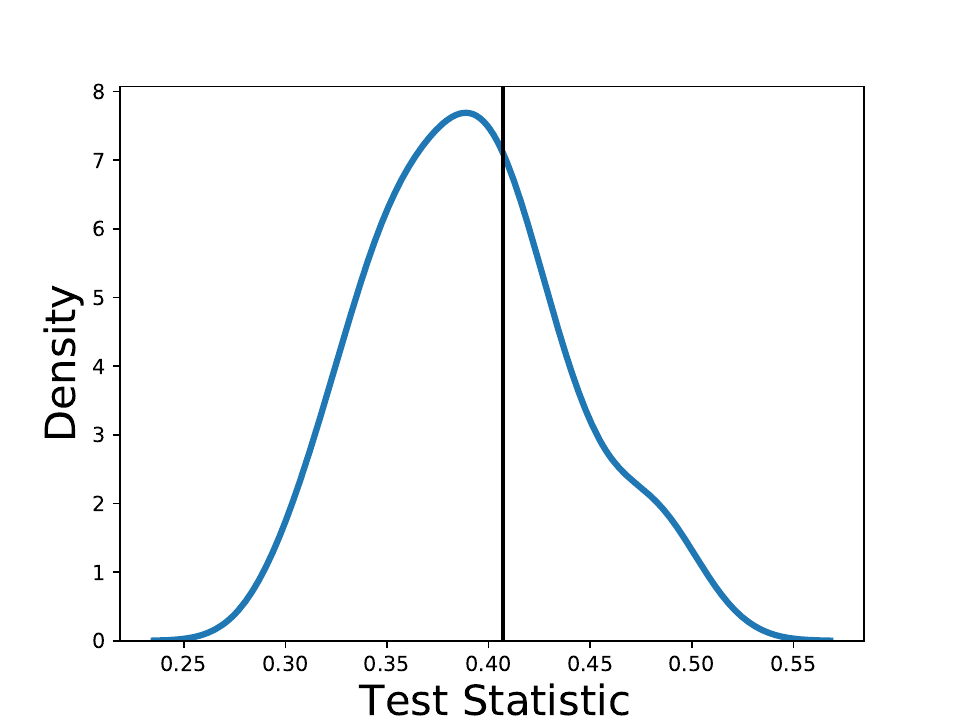}
		
	}
	\caption{Observed test statistic value (black vertical line) and estimated reference density
	for the test statistic based on $100$ simulations for conflict check for the simulated HPV example.}
	\label{hpv-check-sim}
\end{figure}

\section*{Appendix C: Agricultural extensification model}
In this appendix we detail
the model employed in the agricultural extensification example. 
It consists of two modules: a regression model (HM module) used to impute the missing 
Manure Level observations in the archaeological dataset ($M_{\mathcal{A}}$), and a proportional odds model (PO module) used
to specify the
parameter $\gamma$ that is employed to assess the hypothesis of extensification.
The likelihoods and priors of the HM and PO modules are specified below, 
from which both the cut and full variational posteriors can be obtained.

\subsection*{C.1: HM module}
The HM module is a linear Gaussian regression that pools both the archaeological data $\mathcal{A}$ and 
modern data $\mathcal{M}$. The response is
Nitrogen Level ($Z_{d,i}$), and there are fixed effects in Rainfall ($R_{d,i}$) and Manure Level ($M_{d,i}$), as well as a random effect in Site Location ($P_{d,i}$), where $d={\mathcal A},\mathcal{M}$
indexes the dataset
and $i=1,\ldots,n_d$ the observation. The Manure Level is a categorical variable
where $M_{d,i}\in\{m_{low}, m_{med},m_{high}\}$ with 
$m_{low}<m_{med}<m_{high}$. 
Dummy variables $D_{d,i}^{med}=\mathbbm{1}(M_{d,i}=m_{med})$ and 
$D_{d,i}^{high}=\mathbbm{1}(M_{d,i}=m_{high})$
are introduced for medium $m_{med}$
and high $m_{high}$ manuring levels (with low manuring level $m_{low}$ acting
as the baseline category). 
There are $q_2=24$ site locations, so that if $l(d,i)$ is the location of site $P_{d,i}$ then the regression is
\begin{equation}
Z_{d,i}=(1,R_{d,i},D_{d,i}^{med},D_{d,i}^{high})^\top \betavec + \zeta_{l(d,i)}+\epsilon_{d,i}\,,\label{eq:hmmod}
\end{equation}
where $\bm{\beta} = (\beta_1, \dots, \beta_4)^\top$ are fixed effects coefficients, 
and
$\zeta_l\sim \text{N}(0,\sigma^2_\zeta)$ are the $l=1,\ldots,q_2$ location random effect
values. The errors are heteroscedastic, with $\epsilon_{d,i} \sim \text{N}(0, \sigma^2)$ if the $(d,i)$th observation has Crop Category $C_{d,i}$ given by barley, and $\epsilon_{d,i} \sim \text{N}(0, \upsilon\sigma^2)$ if wheat. 

Conditional on the $n_{\mathcal{A}}$ observations on the Manuring Level
$M_{\mathcal A}$ and Rainfall $R_{\mathcal A}$, as well as the 
random effect values $\bm{\zeta}=(\zeta_1,\ldots,\zeta_{q_2})^\top$, 
the likelihood of the HM module is simply given by the product 
of $n_{\mathcal A}+n_{\mathcal M}$ Gaussian densities specified at~\eqref{eq:hmmod}.

Given a large number $g = 1000$, the following priors are used:
\begin{align*}
	& M_{\mathcal{A},i} \sim {\cal U}\left\{1,3\right\}, \quad
	\betavec \sim \text{N}(\bm{0}, gI_4), \quad \sigma^2 \sim \text{U}(0.1,5), \\
	&\sigma_\zeta^2 \sim \text{U}(0.1,5),\quad \widetilde{\upsilon} = \log \upsilon \sim \text{N}(0, g).  
\end{align*}
For $\log R_{\mathcal{A}}$, we use normal priors for the components, where the component-specific means and variances
match those of the uniform priors used in \cite{styring17} (Figure 18 of the supplementary information).   
Define $\widetilde{\sigma} = \log \frac{\sigma^2-0.1}{5-\sigma^2}$, $\widetilde{\sigma_\zeta} = \log \frac{\sigma_\zeta^2-0.1}{5-\sigma_\zeta^2}$ and $\widetilde{\upsilon} = \log \upsilon$, so that $\widetilde{\sigma}$, $\widetilde{\sigma_\zeta}$ and $\widetilde{\upsilon}$ are unconstrained and suitably
approximated by a Gaussian.
Then parameters and unknown quantities in the HM module can be summarized as $\varphivec = (\rhovec^\top, M_{\mathcal{A}}^\top)^\top$, where $\rhovec = (\log(R_{\mathcal{A}})^\top, \betavec^\top, \sigma, \zetavec^\top, \sigma_{\zeta}, \upsilon)^\top$.

\subsection*{C.2: PO Module}
The PO module is a proportional odds model applied only to
the archaeological data with an ordinal response Manure Level ($M_{\mathcal{A},i}$), a fixed effect for Site Size ($S_{\mathcal{A},i}$), random effect for Site Location ($P_{\mathcal{A},i}$), and a logit link function. Let $p_m(P,S) = {\rm Pr}(M \leq m|P,S)$ be the cumulative distribution function of Manure Level $M$ at level $m$
from an observation with Site Size $S$ and Location $P$. Then for the
observations $i = 1, \dots, n_{\mathcal{A}}$,
the proportional odds model has the form
\begin{align*}
\log \left( \frac{p_m(P_{\mathcal{A},i},S_{\mathcal{A},i})}{1 - p_m(P_{\mathcal{A},i},S_{\mathcal{A},i})}   \right) = \alpha_m - \gamma S_{\mathcal{A},i} - \xi_{l(\mathcal{A},i)}\,.
\end{align*}
There are only five locations in the archaeological dataset, 
so that the site location random effect values 
are $\xivec=(\xi_1,\ldots,\xi_5)^\top \sim \text{N}(\bm{0}, \sigma_\xi^2 I_5)$. 
The coefficient
$\alpha_m \in \left\{ \alpha_{low},\alpha_{med} \right\}$ 
varies according to the
manure level value $m$, and $\gamma$ measures the effect of site size. 


Define $\mbox{p}_{m,i}\equiv \mbox{Pr}(M_{\mathcal{A},i}\leq m|P_{\mathcal{A},i},S_{\mathcal{A},i})$, and
the dummy variable
$D_{\mathcal{A},i}^{low}=1-D_{\mathcal{A},i}^{med}-D_{\mathcal{A},i}^{high}$.
Then 
\begin{align*}
p(M_{\mathcal{A}}|\alphavec,\gamma, \xivec, \sigma_\xi) = \prod_{i=1}^{n_{\mathcal{A}}} \left(\mbox{p}_{low,i}\right)^{D_{\mathcal{A},i}^{low}}\left(\mbox{p}_{med,i}-\mbox{p}_{low,i}\right)^{D_{\mathcal{A},i}^{med}}\left(1-\mbox{p}_{med,i}\right)^{D_{\mathcal{A},i}^{high}}.
\end{align*}

Define $\tilde{\alpha}_{med} = \log (\alpha_{med} - \alpha_{low})$ and $\tilde{\sigma_\xi} = \log \frac{\sigma_\xi}{3.5 - \sigma_\xi}$, so that $\tilde{\alpha}_{med}$ and $\tilde{\sigma_\xi}$ are unconstrained and suitably
approximated as Gaussian. Let $\alphavec=(\alpha_{low}, \tilde{\alpha}_{med})^\top$, 
then the parameters in PO module are $\etavec = (\gamma,\alphavec^\top, \xivec^\top, \sigma_\xi)^\top$. The following priors are used:
\[
\gamma \sim \text{N}(0, 4), \quad \alpha_{low} \sim \text{N}(0,1.5), \quad \tilde{\alpha}_{med} \sim \text{N}(-5,7), \quad \sigma_\xi \sim \text{U}(0, 3.5)\,.
\]

\bibliographystyle{chicago}
\bibliography{references}

\end{document}